\documentclass[showpacs,amsmath,amssymb,onecolumn,prx,superscriptaddress,nofootinbib]{revtex4-1}

\usepackage[dvips]{graphicx} 
\usepackage{amsfonts}
\usepackage{amssymb}
\usepackage{amscd}
\usepackage{amsmath}    
\usepackage{enumerate}
\usepackage{epsfig}
\usepackage{subfigure}
\usepackage{xcolor}
\usepackage{amsthm}
\usepackage{slashbox}

\newtheorem{lemma}{Lemma}

\newtheorem{remark}{Remark}

\newcommand{\bra}[1]{\mbox{$\left\langle #1 \right|$}}
\newcommand{\ket}[1]{\mbox{$\left| #1 \right\rangle$}}

\begin{document}

\title{Measurement-device-independent quantum key distribution with uncharacterized qubit sources}

\author{Zhen-Qiang Yin}
\address{Key Laboratory of Quantum Information, University of Science and Technology of China, Hefei 230026, China\\
and Synergetic Innovation Center of Quantum Information $\&$ Quantum Physics, University of Science and Technology of China,\\
Hefei, Anhui 230026, China}
\author{Chi-Hang Fred Fung}
\email{chffung@hku.hk}
\address{Department of Physics and Center of Theoretical and Computational Physics, University of Hong Kong, Pokfulam Road, Hong Kong}
\author{Xiongfeng Ma}
\email{xma@tsinghua.edu.cn}
\address{Center for Quantum Information, Institute for Interdisciplinary Information Sciences, Tsinghua University, Beijing, China}
\author{Chun-Mei Zhang}
\address{Key Laboratory of Quantum Information, University of Science and Technology of China, Hefei 230026, China\\
and Synergetic Innovation Center of Quantum Information $\&$ Quantum Physics, University of Science and Technology of China,\\
Hefei, Anhui 230026, China}
\author{Hong-Wei Li}
\author{Wei Chen}
\email{kooky@mail.ustc.edu.cn}
\author{Shuang Wang}
\email{wshuang@ustc.edu.cn}
\author{Guang-Can Guo}
\affiliation{Key Laboratory of Quantum Information, University of Science and Technology of China, Hefei 230026, China\\
and Synergetic Innovation Center of Quantum Information $\&$ Quantum Physics, University of Science and Technology of China,\\
Hefei, Anhui 230026, China}
\author{Zheng-Fu Han}
\email{zfhan@ustc.edu.cn}
\affiliation{Key Laboratory of Quantum Information, University of Science and Technology of China, Hefei 230026, China\\
and Synergetic Innovation Center of Quantum Information $\&$ Quantum Physics, University of Science and Technology of China,\\
Hefei, Anhui 230026, China}

\begin{abstract}
Measurement-device-independent quantum key distribution (MDIQKD) is proposed to be secure against any possible detection attacks. The security of the original proposal relies on the assumption that the legitimate users can fully characterize the encoding systems including sources. Here, we propose a MDIQKD protocol where we allow uncharacterized encoding systems as long as qubit sources are used. A security proof of the MDIQKD protocol is presented that does not need the knowledge of the encoding states. Simulation results show that the scheme is practical.
\end{abstract}


\maketitle

\section{introduction}
Quantum key distribution (QKD) \cite{Bennett:BB84:1984,Ekert:QKD:1991} allows
two distant parties (Alice and Bob) to share secret key bits. In the most commonly implemented QKD protocol, BB84 \cite{Bennett:BB84:1984}, Alice randomly encodes her qubits into one of the four quantum states, $\ket{0}$, $\ket{1}$, $|+\rangle=(\ket{0}+\ket{1})/\sqrt{2}$, and $|-\rangle=(\ket{0}-\ket{1})/\sqrt{2}$, where $\ket{0}$ and $\ket{1}$ represent the two eigenstates in the $Z$ basis and $|+\rangle$ and $|-\rangle$ represent the two eigenstates in the $X$ basis, respectively. Then, she transmits photons to Bob through a quantum channel which may be controlled by an eavesdropper, Eve. A key can be generated by postprocessing the measurement results \cite{Mayers:Security:2001,Lo:QKDSecurity:1999,Shor:Preskill:2000}. The security of BB84 can be proved by considering the equivalence between BB84 and an entanglement distillation protocol (EDP). Many other QKD protocols are also proven to be secure with similar techniques \cite{Tamaki:B92:2003,Boileau:3state:2005,PhysRevA.82.042335}. Meanwhile, tremendous progress in experimental QKD has also been achieved in the past decade \cite{Zhao:DecoyExp:2006,Rosenberg:ExpDecoy:2007,Zeilinger:ExpDecoy:2007,Peng:ExpDecoy:2007, Zhao:Decoy60km:2006,Yuan:ExpDecoy2007,Takesue:40dBQKD:2007,Stucki:250kmQKD:2009, Wang:260kmQKD:2012,Frolich:QKDnet:2013}. For a review of the subject, one can refer to Ref.~\cite{Scarani:QKDrev:2008} and references therein.

Despite the proven security in theory, Eve may still be able to crack practical QKD systems by exploiting imperfections in actual implementations. Quantum hacking strategies, taking advantage of practical loopholes, were studied in recent years, such as fake-state attack \cite{MAS_Eff_06,Makarov:Fake:08}, time-shift attack \cite{Qi:TimeShift:2007,Zhao:TimeshiftExp:2008}, phase-remapping attack \cite{Fung:Remap:07,Xu:PhaseRemap:2010}, detector-blinding attack \cite{Lydersen:Hacking:2010,Gerhardt:Blind:2011}, and unambiguous state discrimination (USD) attack \cite{Tang:USD:2013}. In order to close all possible loopholes existing in practical QKD systems, device-independent (DI) QKD \cite{MayersYao_98,Acin:DeviceIn:07} protocols have been proposed, whose security does not rely on the details of implementation devices but on the violation of Bell's inequalities or other nonlocality tests. The security is only built on a few reasonable assumptions such as having good random numbers and trusted user operation systems. Unfortunately, DIQKD requires very high detection efficiency and low channel loss to yield secure keys.
A recent DIQKD scheme with local Bell test~\cite{Lim:DIQKD:2013} overcomes the channel loss limitation and extends the achievable distance to 17 km using current experimental parameters.
However, this is still a short distance for practical interest, and,
so far, no DIQKD experiment has been demonstrated.

A close examination on hacking strategies indicates that most loopholes exist in the detection part of QKD systems \cite{Qi:TimeShift:2007,Zhao:TimeshiftExp:2008, MAS_Eff_06,Makarov:Fake:08,Lydersen:Hacking:2010,Gerhardt:Blind:2011}. Along this line, QKD protocols against detection loopholes are proposed \cite{MXF:DDIQKD:2012,Pawlowski:Semi:2011}. Like DIQKD, unfortunately, these protocols still require high performance of implementation devices. In 2012, Lo \emph{et al.} presented their seminal work, measurement-device-independent QKD (MDIQKD) \cite{Lo:MDIQKD:2012}, which is practical and is immune to all possible detector side channel attacks (see also Ref.~\cite{Braunstein:MIQKD:2012}). Recently, a few experimental demonstrations of MDIQKD have been performed \cite{Rubenok:MIQKDexp:2012,daSilva:MIQKD:2012,liu:MIQKDexp:2012,Tang:MDIQKDexp:2013}.
Secret key generation with MDIQKD over a distance of 50 km has been shown~\cite{liu:MIQKDexp:2012}, which is significantly longer than the achievable distance of DIQKD.

In MDIQKD, Alice and Bob each sends their encoded qubits to a measurement unit (MU), controlled by an untrusted party Eve,
who might collaborate with Eve, to perform a Bell-state measurement (BSM). A secure key can be established between Alice and Bob given
Eve's
announcements of the BSM results.
The disadvantage of the original MDIQKD in comparison to DIQKD is that
the security of MDIQKD relies on the assumption that Alice and Bob are able to characterize their encoding systems. For qubit sources, ideal four BB84 encoding states, $\{\ket{0}, \ket{1}, |+\rangle, |-\rangle\}$, are assumed. For coherent state sources, the decoy-state method is assumed \cite{Hwang:Decoy:2003,Lo:Decoy:2005,Wang:Decoy:2005}. In practice, such requirements might not be strictly satisfied. For example, two imperfect encoding systems could be misaligned. In this work, we remove the security assumption that Alice and Bob must characterize their encoding states for MDIQKD. We show that by modifying the original MDIQKD and using qubit sources, one can use uncharacterized encoding systems.
When Alice (Bob) selects to output a state with index $x$ ($y$),
her (his) encoding device emits a mixed qubit state $\rho_{\text{A},x}$ ($\rho_{\text{B},y}$) to Eve.
In our framework, the initial joint state with Eve's system is
$$
\rho_{\text{A},x} \otimes \rho_{\text{B},y} \otimes \rho_\text{E}
\:\:
\text{for all $x,y$}
$$
where Eve's state $\rho_\text{E}$ is independent of $x$ and $y$.
This excludes the case of a joint encoding device where the state is
$\rho_{\text{AB},x,y}\otimes \rho_\text{E}$
and the case of hidden classical variables (unknown to Alice and Bob but known to Eve).
The latter case is often considered as
``black boxes'' in DIQKD papers \cite{MayersYao_98,Acin:DeviceIn:07},
and the state is
$\sum_\lambda p(\lambda)\rho_{\text{A},x}^\lambda \otimes \rho_{\text{B},y}^\lambda \otimes \rho_\text{E}^\lambda$, in which $p(\lambda)$ represents the probability distribution of the hidden variable $\lambda$.
Similarly, the case of pre-shared entanglements with Eve is excluded in our framework.
Therefore, our encoding devices emit mixed qubit states, which are uncharacterized to Alice and Bob. Eve may learn about
the encoding states by only performing quantum operations on them.
There are not other means for her to obtain information.
The additional requirement of devices we put here is that the BSM is able to identify at least two of the four Bell states. Denote the modified protocol as qubit-MDIQKD (QMDIQKD). The QMDIQKD protocol may be particularly relevant to the situation that Alice and Bob can make sure their encoding states are two-dimensional (i.e., the phase encoding system) while they do not trust the accuracy of their phase modulators.


The rest of the paper is organized as follows. In Sec.~\ref{sec-our-protocol}, we present the QMDIQKD protocol with qubit sources, whose security is proven in Sec.~\ref{sec-security-pure-state} for the case of pure-state sources. The security of the general (mixed-) qubit source case is given in Sec.~\ref{sec-mixed-state}. In Sec.~\ref{sec-simulation}, we present simulation results of the QMDIQKD protocol, comparing to the original one. Finally, we conclude in Sec.~\ref{sec-conclusion} with discussions.

\section{Modified MDIQKD protocol} \label{sec-our-protocol}
Similar to the original MDIQKD, Alice and Bob have symmetric encoding systems in the QMDIQKD protocol. They send their encoded qubits to
Eve
for BSM, as shown in Fig.~\ref{Fig:MIsource:MIdiag}. Here, without loss of generality, we assume a phase encoding scheme is used. Alice (Bob) can choose different phases $\varphi_0$, $\varphi_1$, $\varphi_2$, $\varphi_3$ ($\varphi'_0$, $\varphi'_1$, $\varphi'_2$, $\varphi'_3$) to perform the encoding step. Then
Eve
 performs a BSM on the incoming states and announces results to Alice and Bob.

\begin{figure}[hbt]
\centering
\resizebox{12cm}{!}{\includegraphics{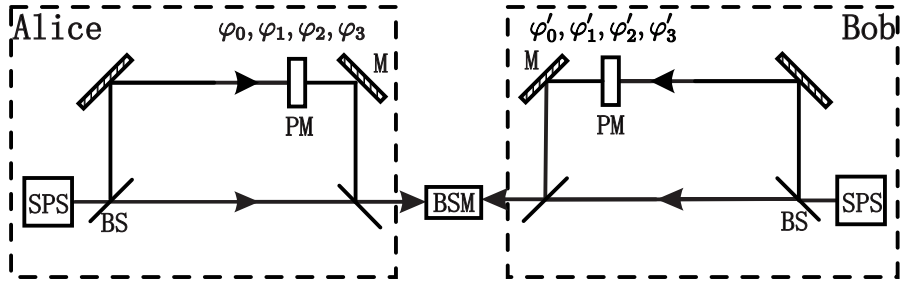}}
\caption{A schematic diagram for the QMDIQKD protocol. SPS: single photon source; BS: beam-splitter; M: mirror; PM: phase modulator; BSM: Bell-state measurement, which is an untrusted device and may be controlled by Eve; $\varphi_0$, $\varphi_1$, $\varphi_2$, $\varphi_3$ ($\varphi'_0$, $\varphi'_1$, $\varphi'_2$, $\varphi'_3$) represent Alice's (Bob's) four encoding choices.}
\label{Fig:MIsource:MIdiag}
\end{figure}

We plan to analyze its security with the EDP method \cite{Lo:QKDSecurity:1999,Shor:Preskill:2000}, which is widely used for security proofs of QKD. The essence of this method is to convert the QMDIQKD protocol into an equivalent EDP, then obtain the relation between the phase error rate and the bit error rate, which can be estimated in experiments. In this section, we give the equivalent EDP version of QMDIQKD and focus on a description where the encoding states are pure (extension to the mixed state description is trivial):

\begin{enumerate}
\item
Alice and Bob prepare $N$ pairs of entangled states,
\begin{equation} \label{MIs:Pure:iniState}
\begin{aligned}
\ket{\phi^+}_{AC} &=\ket{0}_A|\varphi_0\rangle_C+\ket{1}_A|\varphi_1\rangle_C +|2\rangle_A|\varphi_2\rangle_C+|3\rangle_A|\varphi_3\rangle_C, \\
\ket{\phi^+}_{BD} &=\ket{0}_B|\varphi'_0\rangle_D+\ket{1}_B|\varphi'_1\rangle_D +|2\rangle_B|\varphi'_2\rangle_D+|3\rangle_B|\varphi'_3\rangle_D, \\
\end{aligned}
\end{equation}
respectively, where normalization factors are omitted. Subscripts $A$ and $B$ denote Alice's and Bob's classical key bits, respectively, of which values $0$ and $1$ represent encodings in basis $0$, and values $2$ and $3$ represent encodings in basis $1$. The states, $|\varphi_x\rangle_C$ and $|\varphi'_x\rangle_D$ ($x=0,1,2,3$), are Alice's and Bob's uncharacterized encoding states, respectively, to be sent to the MU and in general they are not orthogonal to each other. Alice and Bob ensure that $|\varphi_x\rangle_C$ and $|\varphi'_x\rangle_D$ are fixed pure states in a 2-dimensional Hilbert (qubit) space but do not know the details. Later, we will extend the results to general mixed qubit systems. Here, we describe Alice and Bob's the encoding systems in an equivalent measurement-based way. By measuring her only half of the system, Alice collapses the system $C$ to one of $\ket{\varphi_i}_C$ with $i=0,1,2,3$ with equal probabilities, which is equivalent to saying that Alice prepares the system $C$ into four states with equal probabilities. The same argument holds for Bob. Hence, Eq.~\eqref{MIs:Pure:iniState} shows a mathematical equivalent description of Alice's and Bob¡¯s encoding systems.

\item
Alice and Bob send the states, labeled by $C$ and $D$, respectively, to Eve who announces her BSM result. There are three possible outcomes: BSM failure; a successful measurement result in either of the Bell states
\begin{eqnarray}
|\phi^+\rangle_{CD} &=& (\ket{0}_C\ket{0}_D+\ket{1}_C\ket{1}_D)/\sqrt{2}, \label{MIs:Pure:BellState1} \\
|\psi^+\rangle_{CD} &=& (\ket{0}_C\ket{1}_D+\ket{1}_C\ket{0}_D)/\sqrt{2}. \label{MIs:Pure:BellState2}
\end{eqnarray}
One must note that Eve might not honestly announce her measurement results. The beauty of MDIQKD \cite{Lo:MDIQKD:2012} is that the security does not depend on whether the measurement results are faithfully obtained by Alice and Bob. Another important point is that it is enough to assume that the MU only needs to identify one of the Bell states in the original MDIQKD, whereas in the QMDIQKD, at least two Bell states need to be distinguished.

\item
After receiving Eve's message, Alice and Bob perform bit sift: discarding bits when Eve announces failed BSM ($z=0$). Then, they project systems $A$ and $B$ in Eq.~\eqref{MIs:Pure:iniState} onto $\ket{0}\bra{0}+\ket{1}\bra{1}, \ket{2}\bra{2}+\ket{3}\bra{3}$, which correspond to bases $0$ and $1$, respectively. Then, they perform basis sift: discarding bits when their systems collapse into different bases. By sacrificing certain bits for error testing,\footnote{An alternative way to do that is by performing error verification after error correction \cite{Fung:Finite:2010,MXF:Finite:2011}.} they can then deduce a conditional probability distribution $p(z|x,y)$ [where $z=0,1,2$ stands for Eve's announcements] failure, Eqs.~\eqref{MIs:Pure:BellState1} or ~\eqref{MIs:Pure:BellState2}, respectively; $x$ and $y$ ($x,y\in\{0,1,2,3\}$) represent the states of systems $A$ and $B$.

\item
Finally, Alice and Bob perform EDP to systems $A$ and $B$ and obtain maximally entangled Bell states $|\phi^{+\theta}\rangle_{AB}=(\ket{0}_A\ket{0}_B+e^{i\theta}\ket{1}_A\ket{1}_B)/\sqrt{2}$, where secure key bits can be extracted.
\end{enumerate}

Let us first take a look at the original MDIQKD case where Alice and Bob each sends four BB84 states with equal probabilities. The conditional probabilities of the measurement result by a lossless MU are listed in Table \ref{Tab:MIs:prob}, from where one can see that there is a 50\% intrinsic loss for the original MDIQKD scheme when two Bell states can be distinguished. In the case of $xy=01, 10$ and $z=2$, a bit flip is operated on either Alice's or Bob's side.

\begin{table}[htb]
\centering
\caption{List of conditional probabilities $p(z|x,y)$ for the case where Alice and Bob choose four BB84 states with equal probabilities. Only the cases where they choose the same basis are considered. No loss is considered.}
\label{Tab:MIs:prob}
\begin{tabular}{c|cccc|cccc}
\backslashbox {$z$}{$x,y$}  & 0,0 & 0,1 & 1,0 & 1,1 & 2,2 & 2,3 & 3,2 & 3,3 \\
\hline
0 & 1/2 & 1/2 & 1/2 & 1/2 & 0 & 1 & 1 & 0 \\
1 & 1/2 & 0 & 0 & 1/2 & 1/2 & 0 & 0 & 1/2 \\
2 & 0 & 1/2 & 1/2 & 0 & 1/2 & 0 & 0 & 1/2 \\
\end{tabular}
\end{table}

In step 2, it is crucial in the QMDIQKD protocol to assume that the MU can distinguish at least two Bell states. In the following, we will show that the MDIQKD protocol using the uncharacterized qubit encoding system is insecure if only one Bell state can be measured by the MU, or $z\in\{0,1\}$. In this case, the row of $z=2$ in Table \ref{Tab:MIs:prob} will be emerged to the one of $z=0$. Now, let us consider the following two scenarios.
\begin{enumerate}
\item
Alice's and Bob's encoding systems emit four perfect BB84 states. Then, in the bit-sift procedure, Alice and Bob will discard the cases of $z=0$ and $2$ as listed in Table \ref{Tab:MIs:prob}. Thus, the key bits remain only when Alice and Bob send out the same states.

\item
Alice and Bob's encoding systems emit states from a set of two orthogonal states $\ket{0}$ and $\ket{1}$ regardless of their basis choices. Obviously, no secure key can be generated from this case, since Eve can simply perform a projection measurement on the qubits sent out by Alice and Bob to get the full information.
\end{enumerate}

From Alice's and Bob's points of view, they can not distinguish above two scenarios from their observed results. Thus, such a scheme is not secure. It is not hard to verify that the same conclusion holds when the MU projects onto any other Bell-state instead of Eq.~\eqref{MIs:Pure:BellState1}.

Moreover, note that not any two of Bell states will work for the QMDIQKD protocol. From the attack mentioned above, one can see that it is critical for a joint state of Alice and Bob to be able to yield two Bell state measurement results when they choose the same basis. Table~\ref{table-all-bell-states} shows such possibilities for each case of $x, y$ and four Bell states. One can see that neither the pair $\ket{00+11}$ and $\ket{01-10}$ nor the pair $\ket{01+10}$ and $\ket{00-11}$ can be used for the QMDIQKD protocol.

\begin{table}[h]
\caption{The relation between the states depending on $x,y$ values and four Bell states: $\checkmark$ means there is overlap between the state $xy$ and the corresponding Bell states, while $\times$ means the two states are orthogonal. } \label{table-all-bell-states}
\begin{tabular}{c|c|c|c|c}
\backslashbox {x, y}{Bell state}  & $\ket{00+11}$ & $\ket{01-10}$ & $\ket{01+10}$ & $\ket{00-11}$ \\
\hline
{0,0} & $\checkmark$ & $\times$ & $\times$ & $\checkmark$ \\
{1,1} & $\checkmark$ & $\times$ & $\times$ & $\checkmark$ \\
{0,1} & $\times$ & $\checkmark$ & $\checkmark$ & $\times$ \\
{1,0} & $\times$ & $\checkmark$  & $\checkmark$ & $\times$  \\
\hline
{2,2} & $\checkmark$ & $\times$ & $\checkmark$ & $\times$  \\
{3,3} & $\checkmark$ & $\times$ & $\checkmark$ & $\times$   \\
{2,3} & $\times$ & $\checkmark$ & $\times$ & $\checkmark$   \\
{3,2} & $\times$ & $\checkmark$  & $\times$ & $\checkmark$  \\
\end{tabular}
\end{table}


%
%

\section{Security against collective attacks and key-rate lower bound: pure-state case} \label{sec-security-pure-state}
Following the similar argument used in Ref.~\cite{Shor:Preskill:2000}, we design an EDP that is equivalent to the QMDIQKD protocol. Given that the initial states of Alice's, Bob's, and Eve's ancillas are separable, we first define Eve's collective attack by a unitary transformation:
\begin{equation}
\label{eqn-model1}
\begin{aligned}
&U_{Eve}|\varphi_{x}\rangle_{C}|\varphi'_{y}\rangle_{D}
|e\rangle_{Ea}\ket{0}_M=\sqrt{p(0|x,y)}|\Gamma_{xy0}\rangle_E\ket{0}_M +\sqrt{p(1|x,y)}|\Gamma_{xy1}\rangle_E\ket{1}_M +\sqrt{p(2|x,y)}|\Gamma_{xy2}\rangle_E|2\rangle_M,
\end{aligned}
\end{equation}
where $x,y=\{0,1,2,3\}$, $|e\rangle_{Ea}$ is Eve's arbitrary ancilla, $\ket{0}_M$ is the message which will be sent to Alice and Bob, and $|\Gamma_{xy0}\rangle_{E}$, $|\Gamma_{xy1}\rangle_{E}$ and $|\Gamma_{xy2}\rangle_{E}$ are all normalized Eve's arbitrary quantum states for Eve's ancilla and photons $C$, $D$. Here, we emphasize that our collective attack modeled by Eq.~\eqref{eqn-model1} has taken the basis-independent attack and basis-dependent attack \cite{GLLP:2004} into account. The basis-independent attack is a situation where the density matrix of the states transmitted in the channel when Alice and Bob's bases choices are both $0$ is the same as the case of basis $1$. The basis-dependent attack is a situation where the two density matrices are different.
The basis dependence can be measured by fidelity. In general, for basis-dependent attacks, Eve may obtain basis information through certain measurements and then adopt different operations accordingly. Any measurement that Eve may utilize to learn the basis and the followup operations can be seen as part of an extended unitary transformation on
Alice's and Bob's encoding states and her ancilla given by Eq.~\eqref{eqn-model1}.

\subsection{Example: General four-dimensional case}
\label{sec-general-4-dim}
Before proceeding, we first prove that when two-dimensional photon pairs $C$ and $D$ are general four-dimensional quantum states $|\varphi_{xy}\rangle_{CD}$, the protocol will not be secure. Generally, we assume when Alice's input is $x$ and Bob's input is $y$, the state is $|\varphi_{xy}\rangle_{CD}$. Let us consider a counter example, where
\begin{equation} \label{MIs:Pure:Counter}
\begin{aligned}
\langle\varphi_{00}|\varphi_{11}\rangle_{CD}&=0, \\
\langle\varphi_{00}|\varphi_{01}\rangle_{CD}&=-1/\sqrt{2}, \\
|\varphi_{01}\rangle_{CD}&=|\varphi_{10}\rangle_{CD}, \\
|\varphi_{22}\rangle_{CD}=|\varphi_{33}\rangle_{CD} &=|\varphi_{00}\rangle_{CD}+|\varphi_{01}\rangle_{CD}, \\
|\varphi_{23}\rangle_{CD}&=|\varphi_{32}\rangle_{CD}, \\
\langle\varphi_{23}|\varphi_{00}\rangle_{CD}=\langle\varphi_{23}|\varphi_{11}\rangle_{CD} &=\langle\varphi_{23}|\varphi_{01}\rangle_{CD}=\langle\varphi_{23}|\varphi_{22}\rangle_{CD}=0.
\end{aligned}
\end{equation}
Obviously, the above states are defined in a four-dimensional Hilbert space. Define a conditional probability $p(z|x,y)$ to be the probability for the BSM's result $z$ conditioned on the values set by Alice and Bob, $x,y$. In the case without eavesdropping, Alice and Bob will verify that their conditional probabilities $p(z|x,y)$ are as shown in Table~\ref{table-MDIQKD-s-joint}. These conditional probabilities can be satisfied by the following attack strategy:
\begin{eqnarray*}
U_{Eve}|\varphi_{xy}\rangle_{CD}\ket{0}_{Ea}\ket{0}_M&=&(|\Gamma_{xy0}\rangle_E\ket{0}_{M}+|\Gamma_{xy1}\rangle_E
\ket{1}_M)/\sqrt{2} \text{ when } x=y=0  \text{ or } x=y=1 \\
U_{Eve}|\varphi_{xy}\rangle_{CD}\ket{0}_{Ea}\ket{0}_M&=&(|\Gamma_{xy0}\rangle_E\ket{0}_{M}+|\Gamma_{xy2}\rangle_E|2\rangle_M)/\sqrt{2} \text{ when } x=1, y=0 \text{ or } x=0, y=1 \\
U_{Eve}|\varphi_{23}\rangle_{CD}\ket{0}_{Ea}\ket{0}_M&=&|\Gamma_{230}\rangle_E\ket{0}_{M} \\
|\Gamma_{000}\rangle_E &=& -|\Gamma_{010}\rangle_E.
\end{eqnarray*}
When Alice and Bob declare their basis choices, Eve can know all key values in basis 0, since $|\Gamma_{001}\rangle_E$ and $|\Gamma_{111}\rangle_E$ can be orthogonal.
One can verify that the same probability tables are obtained when the senders of Alice and Bob emit perfect BB84 states without any channel error.
Therefore, it is not possible for Alice and Bob to distinguish whether their senders transmit genuine BB84 states (in which security can hold) or this set of general four-dimensional quantum states $|\varphi_{xy}\rangle_{CD}$ (in which security cannot hold).
Thus, to avoid this case, we restrict
$|\varphi_{xy}\rangle_{CD}=|\varphi_x\rangle_C|\varphi'_y\rangle_D$.
Before discussing the general case, we first give a simple example for the case that Alice and Bob observe a set of specific probabilities $p(z|x,y)$.

\begin{table}[h]
\caption{QMDIQKD: probability table for basis 0 and 1 for a particular joint sender. }
\label{table-MDIQKD-s-joint}
\begin{tabular}{c|cccc|cccc}
\backslashbox {$z$}{$x,y$}  & 0,0 & 0,1 & 1,0 & 1,1 & 2,2 & 2,3 & 3,2 & 3,3 \\
\hline
0 & 1/2 & 1/2 & 1/2 & 1/2 & 0 & 1 & 1 & 0 \\
1 & 1/2 & 0 & 0 & 1/2 & 1/2 & 0 & 0 & 1/2 \\
2 & 0 & 1/2 & 1/2 & 0 & 1/2 & 0 & 0 & 1/2 \\
\end{tabular}
\end{table}

\subsection{Example: Ideal case}
Let us prove the security of the QMDIQKD protocol by considering an ideal case where no disturbance is allowed in the quantum channels. In this case, Alice and Bob would find that $p(0|0,0)=p(1|0,0)=p(0|1,1)=p(1|1,1)=p(0|0,1)=p(2|0,1)=p(0|1,0)=p(2|1,0)=p(1|2,2)=p(2|2,2)=p(1|3,3)=p(2|3,3)=1/2$, and $p(0|2,3)=p(0|3,2)=1$. Then from Eve's collective attack (described by $U_{Eve}$) point of view, she needs to satisfy the following conditions:
\begin{equation} \label{eqn-2}
\begin{aligned}
&U_{Eve}|\varphi_{0}\rangle_{C}|\varphi'_{0}\rangle_{D}
|e\rangle_{Ea}\ket{0}_M=\frac{1}{\sqrt{2}}|\Gamma_{000}\rangle_E\ket{0}_M+\frac{1}{\sqrt{2}}|\Gamma_{001}\rangle_E\ket{1}_M\\
&U_{Eve}|\varphi_{1}\rangle_{C}|\varphi'_{1}\rangle_{D}
|e\rangle_{Ea}\ket{0}_M=\frac{1}{\sqrt{2}}|\Gamma_{110}\rangle_E\ket{0}_M+\frac{1}{\sqrt{2}}|\Gamma_{111}\rangle_E\ket{1}_M\\
&U_{Eve}|\varphi_{0}\rangle_{C}|\varphi'_{1}\rangle_{D}
|e\rangle_{Ea}\ket{0}_M=\frac{1}{\sqrt{2}}|\Gamma_{010}\rangle_E\ket{0}_M+\frac{1}{\sqrt{2}}|\Gamma_{012}\rangle_E|2\rangle_M\\
&U_{Eve}|\varphi_{1}\rangle_{C}|\varphi'_{0}\rangle_{D}
|e\rangle_{Ea}\ket{0}_M=\frac{1}{\sqrt{2}}|\Gamma_{100}\rangle_E\ket{0}_M+\frac{1}{\sqrt{2}}|\Gamma_{102}\rangle_E|2\rangle_M\\
\end{aligned}
\end{equation}
when Alice and Bob prepare the $0$ basis and
\begin{equation}
\label{eqn-3}
\begin{aligned}
&U_{Eve}|\varphi_{2}\rangle_{C}|\varphi'_{2}\rangle_{D}
|e\rangle_{Ea}\ket{0}_M=\frac{1}{\sqrt{2}}|\Gamma_{221}\rangle_E\ket{1}_M+\frac{1}{\sqrt{2}}|\Gamma_{222}\rangle_E|2\rangle_M\\
&U_{Eve}|\varphi_{3}\rangle_{C}|\varphi'_{3}\rangle_{D}
|e\rangle_{Ea}\ket{0}_M=\frac{1}{\sqrt{2}}|\Gamma_{331}\rangle_E\ket{1}_M+\frac{1}{\sqrt{2}}|\Gamma_{332}\rangle_E|2\rangle_M\\
&U_{Eve}|\varphi_{2}\rangle_{C}|\varphi'_{3}\rangle_{D}
|e\rangle_{Ea}\ket{0}_M=|\Gamma_{230}\rangle_E\ket{0}_M\\
&U_{Eve}|\varphi_{3}\rangle_{C}|\varphi'_{2}\rangle_{D}
|e\rangle_{Ea}\ket{0}_M=|\Gamma_{320}\rangle_E\ket{0}_M\\
\end{aligned}
\end{equation}
when Alice and Bob prepare the $1$ basis. This definition is just a special case of Eq.~\eqref{eqn-model1}.
We again assume that Alice and Bob do not know the details of
$|\varphi_x\rangle_C$ and $|\varphi'_y\rangle_D$ $(x,y=0,1,2,3)$.
Recall that $|\varphi_x\rangle_C$ and $|\varphi'_y\rangle_D$
are both in the two-dimensional Hilbert space and they are disjoint ($|\varphi_{xy}\rangle_{CD}=|\varphi_x\rangle_C|\varphi'_y\rangle_D$)
(see previous Sec.~\ref{sec-general-4-dim}),
and we may arbitrarily assign a phase to each of them. Thus
\begin{equation}
\label{eqn-varphi23}
\begin{aligned}
|\varphi_2\rangle_C&=C_{20}|\varphi_0\rangle_C+C_{21}|\varphi_1\rangle_C \\
|\varphi_3\rangle_C&=C_{30}|\varphi_0\rangle_C+C_{31}e^{i\theta}|\varphi_1\rangle_C \\
|\varphi'_2\rangle_D&=C'_{20}|\varphi'_0\rangle_D+C'_{21}|\varphi'_1\rangle_D \\
|\varphi'_3\rangle_D&=C'_{30}|\varphi'_0\rangle_D+C'_{31}e^{i\theta'}|\varphi'_1\rangle_D
\end{aligned}
\end{equation}
must hold for some complex numbers $C_{xy}$ and $C'_{xy}$. By considering that $|\varphi_0\rangle_C$ and $|\varphi_1\rangle_C$ can have some trivial overall phases, we can assume $C_{20}$ and $C_{21}$ are both non-negative real numbers. Next, by omitting the trivial overall phase of $|\varphi_3\rangle_C$, we can simply assume that $C_{30}$ and $C_{31}$ are also non-negative real numbers. In the same way, $C'_{xy}$ are also non-negative real numbers.  Substitute Eqs. \eqref{eqn-2} to the last two equations of \eqref{eqn-3}, and we obtain
\begin{equation}
\begin{aligned}
C_{30}C'_{20}|\Gamma_{001}\rangle_E+C_{31}C'_{21}e^{i\theta}|\Gamma_{111}\rangle_E&=0\\
C_{30}C'_{21}|\Gamma_{012}\rangle_E+C_{31}C'_{20}e^{i\theta}|\Gamma_{102}\rangle_E&=0\\
C_{20}C'_{30}|\Gamma_{001}\rangle_E+C_{21}C'_{31}e^{i\theta'}|\Gamma_{111}\rangle_E&=0\\
C_{20}C'_{31}|\Gamma_{012}\rangle_E+C_{21}C'_{30}e^{-i\theta'}|\Gamma_{102}\rangle_E&=0.\\
\end{aligned}
\end{equation}
 From the above equation, one can verify that if any one of $\{C_{xy}, C'_{xy}| x=2,3,y=0,1\}$ equals to $0$, Eq. \eqref{eqn-3} and the normalized conditions of $|\varphi_2\rangle_C,|\varphi'_2\rangle_D,|\varphi_3\rangle_C,|\varphi'_3\rangle_D$ cannot be all satisfied. Hence $C_{xy}\not=0$, $C'_{xy}\not=0$, and furthermore we have
\begin{equation}
\label{eqn-5}
\begin{aligned}
&|\Gamma_{001}\rangle_E+e^{i\theta}|\Gamma_{111}\rangle_E=0.\\
\end{aligned}
\end{equation}

Obviously, Eq.~\eqref{eqn-5} makes sure that Eve has no information on Alice's and Bob's bits in basis $0$ when Eve announces the message $\ket{1}_M$. Therefore, the above-observed probabilities can promise Alice and Bob to share secure-key bits even when their encoding states are not characterized. In other words, we have proven the security of the QMDIQKD protocol in the case of no error and no loss.

\subsection{Proof: General pure-qubit case}
To begin our analysis, without loss of generality, we can assume in Eq.~\eqref{eqn-model1} $|\Gamma_{xyz}\rangle_E=\sum_n\gamma_{xyzn}|n\rangle_E$, in which $|n\rangle_E$ are a set of normalized orthogonal bases of Eve's states, and complex number $\gamma_{xyzn}=_E\langle n|\Gamma_{xyz}\rangle_E$, satisfying $\sum_n\big|\gamma_{xyzn}\big|^2=1$. Thus, we can give the density matrix for the case that Alice and Bob both select basis $0$:

\begin{equation}
\label{eqn-rhoAB-1}
\begin{aligned}
\rho=&\frac{1}{p(1|0,0)+p(1|1,1)+p(1|0,1)+p(1|1,0)}\cdot \\
&\sum_nP\{£¨\sqrt{p(1|0,0)}\gamma_{001n}\ket{0}_A\ket{0}_B+\sqrt{p(1|1,1)}\gamma_{111n}\ket{1}_A\ket{1}_B+\sqrt{p(1|0,1)}\gamma_{011n}\ket{0}_A\ket{1}_B+\sqrt{p(1|1,0)}\gamma_{101n}\ket{1}_A\ket{0}_B\},
\end{aligned}
\end{equation}
in which, $P\{|x\rangle\}=|x\rangle\langle x|$. The aim of this EDP is to obtain perfect Bell states $|\phi^{-\theta}\rangle_{AB}=(\ket{0}_A\ket{0}_B-e^{-i\theta}\ket{1}_A\ket{1}_B)/\sqrt{2}$. Accordingly, we can define the bit error rate $e_{b}$ and phase error rate $e_{p}$ under basis $0$:

\begin{eqnarray}
e_b&=& _A\langle 0|_B\langle 1|\rho\ket{1}_B\ket{0}_A+_A\langle 1|_B\langle 0|\rho\ket{0}_B\ket{1}_A=\frac{p(1|0,1)+p(1|1,0)}{p(1|0,0)+p(1|1,1)+p(1|0,1)+p(1|1,0)}\\
e_p&=& _{AB}\langle \phi^{+\theta}|\rho|\phi^{+\theta}\rangle_{AB}+_{AB}\langle \psi^{+\theta}|\rho|\psi^{+\theta}\rangle_{AB}
\label{eqn-eb-ep}
\\
&=&\frac{\sum_n\big|\sqrt{p(1|0,0)}\gamma_{001n}+e^{i\theta}\sqrt{p(1|1,1)}\gamma_{111n}\big|^2+\sum_n\big|e^{i\theta}\sqrt{p(1|0,1)}\gamma_{011n}+\sqrt{p(1|1,0)}\gamma_{101n}\big|^2}{2(p(1|0,0)+p(1|1,1)+p(1|0,1)+p(1|1,0))},
\end{eqnarray}
in which, $|\phi^{+\theta}\rangle_{AB}=(\ket{0}_A\ket{0}_B+e^{-i\theta}\ket{1}_A\ket{1}_B)/\sqrt{2}$, and $|\psi^{+\theta}\rangle_{AB}=(e^{-i\theta}\ket{0}_A\ket{1}_B+\ket{1}_A\ket{0}_B)/\sqrt{2}$. The task is to find a way to estimate an effective upper-bound of $e_p$. Before proceeding, we remark that we just focus on the $e_p$ and final key-bits rate for the $0$ basis for simplicity, and thus we do not need to calculate the density matrix for the $1$ basis. But, the $e_p$ of this $0$ basis must be related to some probabilities of the $1$ basis, such as $p(1|2,3)$. Now, we begin to detail how to obtain an upper bound of $e_p$.


We substitute the relations \eqref{eqn-varphi23} into Eq.~\eqref{eqn-model1} to obtain the following constraints:
\begin{equation}
\label{eqn-constraints1}
\begin{aligned}
&C_{x0}C'_{y0}\sqrt{p(z|0,0)}|\Gamma_{00z}\rangle_E
+C_{x0}C'_{y1}e^{ia\theta'}\sqrt{p(z|0,1)}|\Gamma_{01z}\rangle_E\\
&+C_{x1}C'_{y0}e^{ib\theta}\sqrt{p(z|1,0)}|\Gamma_{10z}\rangle_E
+C_{x1}C'_{y1}e^{i(a\theta'+b\theta)}\sqrt{p(z|1,1)}|\Gamma_{11z}\rangle_E\\&=\sqrt{p(z|x,y)}|\Gamma_{xyz}\rangle_E,
\end{aligned}
\end{equation}
in which, $x,y=\{0,1,2,3\}$, $z=\{0,1,2\}$, $a=0$ when $y\not=3$, $a=1$ when $y=3$, $b=0$ when $x\not=3$, and $b=1$ when $x=3$. Considering that $|\Gamma_{xyz}\rangle_E$ can be spanned by a set of basis $|n\rangle_E$ and with Eq.~\eqref{eqn-constraints1}, we obtain
\begin{equation}
\label{eqn-constraints2}
\begin{aligned}
&\sum_n\big|C_{30}C'_{20}\sqrt{p(1|0,0)}\gamma_{001n}+C_{31}C'_{21}\sqrt{p(1|1,1)}e^{i\theta}\gamma_{111n}\big|^2\\
&\leqslant
(\sqrt{p(1|3,2)}+C_{30}C'_{21}\sqrt{p(1|0,1)}+C_{31}C'_{20}\sqrt{p(1|1,0)})^2\triangleq \xi \\
&\sum_n\big|C_{30}C'_{21}\sqrt{p(2|0,1)}\gamma_{012n}+C_{31}C'_{20}\sqrt{p(2|1,0)}e^{i\theta}\gamma_{102n}\big|^2\\
&\leqslant
(\sqrt{p(2|3,2)}+C_{30}C'_{20}\sqrt{p(2|0,0)}+C_{31}C'_{21}\sqrt{p(2|1,1)})^2\triangleq \zeta \\
\end{aligned}
\end{equation}
Note that $|\varphi_x\rangle_C$ must be normalized, and thus we have
\begin{equation}
\label{eqn-constraints3}
\begin{aligned}
&C^2_{30}+C^2_{31}+2C_{30}C_{31}Re(e^{i\theta}_C\langle \varphi_0|\varphi_1\rangle_C)=1\\
&C'^2_{20}+C'^2_{21}+2C'_{20}C'_{21}Re(_D\langle \varphi'_0|\varphi'_1\rangle_D)=1,\\
\end{aligned}
\end{equation}
where $Re(x)$ represents the real part of complex number $x$. From Eq.~\eqref{eqn-model1}, it is easy to verify that
\begin{equation} \label{MIs:Pure:Real}
\begin{aligned}
\big|Re(e^{i\theta}_C\langle \varphi_0|\varphi_1\rangle_C)\big| &\leqslant \sqrt{p(0|1,0)p(0|0,0)}+\sqrt{p(1|1,0)p(1|0,0)}+\sqrt{p(2|1,0)p(2|0,0)}\triangleq \chi \\
\big|Re(_D\langle \varphi'_0|\varphi'_1\rangle_D)\big| &\leqslant \sqrt{p(0|0,1)p(0|0,0)}+\sqrt{p(1|0,1)p(1|0,0)}+\sqrt{p(2|0,1)p(2|0,0)}\triangleq \chi'.
\end{aligned}
\end{equation}
 From constraints~\eqref{eqn-constraints2}, we also know that

\begin{equation}
\label{eqn-constraints4}
\begin{aligned}
&(C_{30}C'_{20}\sqrt{p(1|0,0)}-C_{31}C'_{21}\sqrt{p(1|1,1)})^2\leqslant \xi\\
&(C_{30}C'_{21}\sqrt{p(2|0,1)}-C_{31}C'_{20}\sqrt{p(2|1,0)})^2\leqslant \zeta,
\end{aligned}
\end{equation}
easily. Assuming that one obtains $C_{30}$,$C_{31}$, $C'_{20}$, and $C'_{21}$ satisfying the above constraints, the only remaining task is to find the maximum of $\sum_n\big|\sqrt{p(1|0,0)}\gamma_{001n}+\sqrt{p(1|1,1)}e^{i\theta}\gamma_{111n}\big|^2$. By observing the first inequality of \eqref{eqn-constraints2} and with the help of triangle inequality
and Cauchy-Schwarz inequality, we have:

\begin{equation}
\begin{aligned}
\xi&\geqslant \sum_n (C_{30}C'_{20}\big|\sqrt{p(1|0,0)}\gamma_{001n}+\sqrt{p(1|1,1)}e^{i\theta}\gamma_{111n}\big|-\big|C_{30}C'_{20}-C_{31}C'_{21}\big|\sqrt{p(1|1,1)}\big|\gamma_{111n}\big|)^2\\
&=C^2_{30}C'^2_{20}\sum_n\big|\sqrt{p(1|0,0)}\gamma_{001n}+e^{i\theta}\sqrt{p(1|1,1)}\gamma_{111n}\big|^2+(C_{30}C'_{20}-C_{31}C'_{21})^2p(1|1,1)\\
&-2C_{30}C'_{20}\big|C_{30}C'_{20}-C_{31}C'_{21}\big|\sqrt{p(1|1,1)}\sum_n\big|\sqrt{p(1|0,0)}\gamma_{001n}+e^{i\theta}\sqrt{p(1|1,1)}\gamma_{111n}\big|\big|\gamma_{111n}\big|\\
&\geqslant C^2_{30}C'^2_{20}\sum_n\big|\sqrt{p(1|0,0)}\gamma_{001n}+\sqrt{p(1|1,1)}e^{i\theta}\gamma_{111n}\big|^2+(C_{30}C'_{20}-C_{31}C'_{21})^2p(1|1,1)\\
&-2C_{30}C'_{20}\big|C_{30}C'_{20}-C_{31}C'_{21}\big|\sqrt{p(1|1,1)}\sqrt{\sum_n\big|\sqrt{p(1|0,0)}\gamma_{001n}+e^{i\theta}\sqrt{p(1|1,1)}\gamma_{111n}\big|^2}\\
&=(C_{30}C'_{20}\sqrt{\sum_n\big|\sqrt{p(1|0,0)}\gamma_{001n}+\sqrt{p(1|1,1)}e^{i\theta}\gamma_{111n}\big|^2}-\big|C_{30}C'_{20}-C_{31}C'_{21}\big|\sqrt{p(1|1,1)})^2.\\
\end{aligned}
\end{equation}
Therefore, we obtain
\begin{equation}
\label{eqn-max-exp1}
\begin{aligned}
&\sum_n\big|\sqrt{p(1|0,0)}\gamma_{001n}+\sqrt{p(1|1,1)}e^{i\theta}\gamma_{111n}\big|^2\leqslant \frac{\big(\sqrt{\xi}+\big|C_{30}C'_{20}-C_{31}C'_{21}\big|\sqrt{p(1|1,1)}\big)^2}{C^2_{30}C'^2_{20}},
\end{aligned}
\end{equation}
Furthermore, by the same way, we can obtain that
\begin{equation}
\label{eqn-max-exp2}
\begin{aligned}
&\sum_n\big|\sqrt{p(1|0,0)}\gamma_{001n}+\sqrt{p(1|1,1)}e^{i\theta}\gamma_{111n}\big|^2\leqslant \frac{\big(\sqrt{\xi}+\big|C_{30}C'_{20}-C_{31}C'_{21}\big|\sqrt{p(1|0,0)}\big)^2}{C^2_{31}C'^2_{21}}.
\end{aligned}
\end{equation}
Combining constraints \eqref{eqn-max-exp1} and \eqref{eqn-max-exp2}, we have
\begin{equation}
\label{eqn-max-exp}
\begin{aligned}
&\sum_n\big|\sqrt{p(1|0,0)}\gamma_{001n}+\sqrt{p(1|1,1)}e^{i\theta}\gamma_{111n}\big|^2\\&
\leqslant min\{max_{C,C'}\{\frac{\big(\sqrt{\xi}+\big|C_{30}C'_{20}-C_{31}C'_{21}\big|\sqrt{p(1|1,1)}\big)^2}{C^2_{30}C'^2_{20}}\},\\
&max_{C,C'}\{\frac{\big(\sqrt{\xi}+\big|C_{30}C'_{20}-C_{31}C'_{21}\big|\sqrt{p(1|0,0)}\big)^2}{C^2_{31}C'^2_{21}}\}\}\triangleq \varepsilon,
\end{aligned}
\end{equation}
in which $max_{C,C'}$ means searching over all $C_{30}$, $C_{31}$, $C'_{20}$ and $C'_{21}$ satisfying constraints \eqref{eqn-constraints3} and \eqref{eqn-constraints4}, and $max\{a,b\}$ ($min\{a,b\}$) yields the larger (smaller) one of $a$ and $b$. Although finding an analytical expression of $\varepsilon$ may be difficult, one can get $\varepsilon$ with numerical methods. When $\varepsilon$ is obtained, the phase error rate is given by:
\begin{equation}
\label{eqn-ep-bound}
\begin{aligned}
e_{p}\leqslant \frac{(\sqrt{p(1|0,1)}+\sqrt{p(1|1,0)})^2+\varepsilon}{2(p(1|0,0)+p(1|1,1)+p(1|0,1)+p(1|1,0))}.
\end{aligned}
\end{equation}
One should note that Eq.~\eqref{eqn-ep-bound} requires that $C_{30}C'_{20}$ and $C_{31}C'_{21}$ cannot be both zero. If the constraints \eqref{eqn-constraints3} and \eqref{eqn-constraints4} allow $C_{30}C'_{20}$ and $C_{31}C'_{21}$ to be both zero, $e_p$ can take on an arbitrary value~\footnote{This situation can occur when $|\varphi_0\rangle=|\varphi_2\rangle=|0\rangle$ and $|\varphi_1\rangle=|\varphi_3\rangle=|1\rangle$ for both Alice and Bob.  In this case, security is obviously impossible, which is consistent with that both $C_{30}C'_{20}$ and $C_{31}C'_{21}$ can be zero.}. The secure key rate is given by
\begin{equation}
\label{eqn-key-rate-formula1}
R=1-H(e_b)-H(e_p),
\end{equation}
where $H(x)=-x \log x -(1-x) \log (1-x)$ is the Shannon's binary entropy function.
Based on this lower bound of the secret key rates,
in Sec.~\ref{sec-simulation}, we will show a numerical simulation of real-life QMDIQKD.
But before that, we argue in the next section that the secret key rate obtained here for the pure-state case is applicable to the mixed-state case as well.

\section{Security Proof for the mixed-state case}
\label{sec-mixed-state}

We proved the security of our QMDIQKD scheme when Alice and Bob send pure qubit states in Sec.~\ref{sec-security-pure-state}.
In this case, the key rate is given by Eq.~\eqref{eqn-key-rate-formula1} with $e_p$ bounded by Eq.~\eqref{eqn-ep-bound}.
Here, we prove below that the exact same key rate formula and $e_p$ bound are applicable to the case where Alice and Bob send mixed qubit states. The proof strategy is to reuse the pure-state results by relying on linearity.

We first purify the mixed states of Alice and Bob in a system $F$.
Then the entire state of Alice and Bob can be conditional on value $i$ of system $F$, and we denote it as $\rho_{ABCD,i}$.
Note that for each $i$, Alice's and Bob's states in $\rho_{ABCD,i}$ are pure.
The channel by Eve converts the state $\rho_{ABCD,i}$ to a final state in systems $A$ and $B$ as
in Eq.~\eqref{eqn-rhoAB-1}:
$$
\rho_{ABCD,i} \longrightarrow \rho_{AB,i} = \mathcal{E}(\rho_{ABCD,i})
$$
where $\rho_{AB,i}$ is the expression in Eq.~\eqref{eqn-rhoAB-1}, and $\mathcal{E}$ represents the channel.
For simplicity we denote $\rho_i=\rho_{ABCD,i}$.
We may calculate $e_{p,i}$ as
in Eq.~\eqref{eqn-eb-ep}
for each $i$ of system $F$, which we denote as follows:
$$
e_{p,i}=g'(\rho_{AB,i}),
$$
where $g'(\rho)$ is defined as the right-hand side of Eq.~\eqref{eqn-eb-ep}.
Note that $g'$ is linear.
Furthermore, $g'$ only depends on the measurement statistics of $\rho_{AB,i}$, denoted as $s_{i,j}=\operatorname{Tr} (O_j \rho_{AB,i})$ for some measurement $O_j$ (e.g., projection onto a Bell state $|\phi^{+\theta}\rangle$).
Let $\vec{s}(\rho_{AB,i})=[s_{i,1}, s_{i,2},\ldots]$ denote the operation that returns a collection of measurement statistics.
Thus, we can express $e_{p,i}$ as a function of the statistics instead:
\begin{equation}
\label{eqn-mixed-state-epi}
e_{p,i}=g( \vec{s}( \mathcal{E} ( \rho_i ) ) ),
\end{equation}
where $g \circ \vec{s} = g'$.

Denote the key rate taking into account privacy amplification only by $R_\text{PA}(e_p)=1-H(e_p)$.

\begin{remark}
{\rm
$R_\text{PA}(e_{p,i})$ is a secure key rate since for each $i$,
Alice's and Bob's states in $\rho_{ABCD,i}$ are pure and thus the key rate formula and the bound for the phase error rate in Sec.~\ref{sec-security-pure-state} apply.
}
\end{remark}

\begin{remark}
{\rm
$\sum_i p_i R_\text{PA}(e_{p,i})$ is a secure key rate because Alice and Bob could use the value of system $F$ to compute a key for each $i$.
}
\end{remark}

\begin{lemma}
\label{lemma-1}
{\rm
Given that $\sum_i p_i R_\text{PA}(e_{p,i})$ is a secure key rate,
$R_\text{PA}(\sum_i p_i e_{p,i})$ is a secure key rate.
This corresponds to the case the Alice and Bob do not use the value of system $F$,
and so only the average phase error rate
$\sum_i p_i e_{p,i}$ is used.
}
\end{lemma}
\begin{proof}
This is true because of the convexity of $R_\text{PA}(\cdot)$ in the domain $[0,1/2]$.
\end{proof}

\begin{remark}
{\rm
We express Eq.~\eqref{eqn-ep-bound} in a generic manner:
\begin{align}
g(\vec{s}) \leq f(\vec{s})
\text{ for any collection of measurement statistics } \vec{s},
\label{eqn-gf}
\end{align}
where $g(\vec{s})$ represents $e_p$ generated by some measurement statistics $\vec{s}$ (c.f. Eqs.~\eqref{eqn-mixed-state-epi} and \eqref{eqn-eb-ep})
and $f(\vec{s})$ represents the right-hand side of Eq.~\eqref{eqn-ep-bound}.
Note that $\vec{s}$ contains $p(1|0,0)$, $p(1|0,1)$, for example.
}
\end{remark}

\begin{lemma}
{\rm
$R_\text{PA}(f( \vec{s}( \mathcal{E} ( \rho ) ) ))$ is a secure key rate, where $\rho=\sum_i p_i \rho_i$ is the average state obtained by ignoring system $F$.
}
\end{lemma}
\begin{proof}
Note that Lemma~\ref{lemma-1} shows that
\begin{align}
R_\text{PA}(\sum_i p_i e_{p,i})
\end{align}
is secure.
(This corresponds to giving system $F$ to Eve.  But this point is not relevant to our current discussion.)
Three facts are important: $\mathcal{E}$ is linear, $g$ is linear because $g'$ is linear, and $\vec{s}$ is linear.
This means that the parameter in the above equation is
\begin{eqnarray*}
\sum_i p_i e_{p,i}&=&\sum_i p_i g( \vec{s}( \mathcal{E} ( \rho_i ) ) )
\\
&= &
g( \vec{s}( \mathcal{E} ( \sum_i p_i \rho_i ) ) )
\\
&\leq&
f( \vec{s}( \mathcal{E} ( \sum_i p_i \rho_i ) ) )
%
\end{eqnarray*}
where the last line is due to Eq.~\eqref{eqn-gf}.

Now, note that $R_\text{PA}(x)$ is a decreasing function of $x$, i.e., $R_\text{PA}(x) \geq R_\text{PA}(y)$ for $y \geq x$ (in the domain $[0,1/2]$).
Thus,
\begin{align}
R_\text{PA}(\sum_i p_i e_{p,i})
\geq
R_\text{PA}(
f( \vec{s}( \mathcal{E} ( \sum_i p_i \rho_i ) ) )
).
\end{align}
Since the left-hand side is secure, the right-hand side must represent a secure key rate when Alice and Bob ignore system $F$ and use only the average state $\rho=\sum_i p_i \rho_i$ which contains the encoding state as a mixed state.
\end{proof}
This proves that
bounding the phase error rate by Eq.~\eqref{eqn-ep-bound} with
Alice and Bob sending mixed qubit states produces a secure key rate $R_\text{PA}$.

For the error correction part of the key rate formula $R_\text{EC}(e_b)\triangleq-H(e_b)$, it can be easily seen that it is applicable to the mixed-state case as well since $e_b$ is directly measured.
Alternatively,
the same line of arguments in the above analysis for $e_p$ could be used for $e_b$ as well, showing that $R_\text{EC}$ is secure.
In essence, we have shown that
the key rate given by Eq.~\eqref{eqn-key-rate-formula1} with $e_p$ bounded by Eq.~\eqref{eqn-ep-bound}
is applicable to the case where Alice and Bob send mixed-qubit states.

\section{Simulation}
\label{sec-simulation}

 We first consider the case that Alice and Bob use ideal BB84 senders (not trusted by Alice and Bob) and an ideal Bell-state MU to perform QMDIQKD with noise-free channel and no Eve's attack, but with photon absorption taken into account. One must observe that $p(1|2,3)=p(1|0,1)=p(2|0,0)=p(2|1,1)=p(2|2,3)=0$, then with
constraint \eqref{eqn-constraints2} we deduce that $\xi=\zeta=0$.
Then, through constraints \eqref{eqn-constraints3} and \eqref{eqn-constraints4}, it is easy to verify that $C_{30}=C_{31}=C'_{20}=C'_{21}\not=0$. Hence, $\varepsilon=0$ and then  $e_p=0$. Therefore, in this case, Alice and Bob can share perfect Bell state $|\phi^{-\theta}\rangle_{AB}$. Furthermore, MDIQKD can be secure in this situation, which fits in well with the example given in Sec.~\ref{sec-our-protocol}.

 For a general situation, an analytical expression is hard to obtain. However, according to $p(z|x,y)$ which is directly known from experiments, $\chi$ can be
deduced. Then all $C_{30}$, $C_{31}$, $C'_{20}$, and $C'_{21}$ satisfying constraints \eqref{eqn-constraints3} and \eqref{eqn-constraints4} can be searched, and the maximum of $\varepsilon$ can be obtained according to Eq.~\eqref{eqn-max-exp}.
Finally,
an upper bound of $e_{p}$ can be obtained.

 To estimate the performance of QMDIQKD, a numeric simulation is given. For comparison, we assume that we use perfect BB84 senders and Bell-state MU (but we do not trust them) to perform the QMDIQKD protocol. The MU, whose implementation can be imagined as the same as the one in Ref. \cite{Lo:MDIQKD:2012}, is equipped with four single photon detectors (SPDs) with the detection efficiency $\eta$ and dark counting rate $d$ $per$ $gate$. We assume that Eve is passive and ignore the channel disturbance and optical misalignment. We also define that $l$ represents the distance from Alice or Bob to the MU and $p_s=10^{-0.02l}\eta$ means the probability that a single photon from Alice or Bob can click a SPD of MU.

 Consequently, if Alice and Bob both emit qubit $|0\rangle$ or $|1\rangle$, the MU will announce message $1$ with probability $p(1|0,0)=p(1|1,1)=p_s^2(1-d)^2/2+2p_s(1-p_s)d(1-d)^2+2(1-p_s)^2d^2(1-d)^2$, in which the first item corresponds to the case that the projection of the incoming photons into $|\phi^+\rangle_{CD}$ is successful: the two photons click the two SPDs and the remaining two SPDs do not give dark clicks. The second item accounts for the case that only one photon clicks one SPD but a dark count occurs in one relevant SPD, and the last item represents the case that two photons are absorbed by the channel but two dark counts occur in two relevant SPDs. If Alice and Bob both emit qubit $|0\rangle$ or $|1\rangle$, the MU will announce message $2$ with probability $p(2|0,0)=p(2|1,1)=2(1-p)^2d^2(1-d)^2+2p(1-p)d(1-d)^2$, in which the first item accounts for the case that only one photon clicks one SPD but a dark count occurs in one relevant SPD, and the last item represents the case that two photons are absorbed by the channel but two dark counts occur in two relevant SPDs. And $p(0|0,0)=p(0|1,1)=1-p(1|0,0)-p(2|0,0)$ also holds.

 By the similar considerations, we set $p(1|1,0)=p(1|0,1)=2(1-p_s)^2d^2(1-d)^2+2p_s(1-p_s)d(1-d)^2$, $p(2|0,1)=p(2|1,0)=p_s^2(1-d)^2/2+2p_s(1-p_s)d(1-d)^2+2(1-p_s)^2d^2(1-d)^2$, $p(0|0,1)=p(0|1,0)=1-p(1|0,1)-p(2|0,1)$, $p(1|2,3)=p(1|3,2)=p(2|2,3)=p(2|3,2)=2(1-p_s)^2d^2(1-d)^2+2p_s(1-p_s)d(1-d)^2$, and $p(0|2,3)=p(0|3,2)=1-p(1|2,3)-p(2|2,3)$. The secure-key rate (unit: per sifted key bit under $0$ basis) versus channel distance $L$(km) from Alice or Bob to the MU is given by Fig.~\ref{fig-1}.

 \begin{figure}[!t]
{\includegraphics[width=.5\columnwidth]{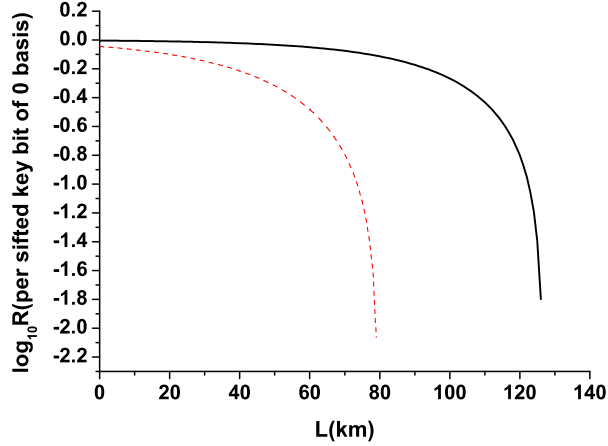}}
\caption{Secure-key rate $R$ (unit: per sifted key bit of $0$ basis) vs channel distance $L$ (km) from Alice or Bob to MU: we set $\eta=0.1$, $d=10^{-5}$ per pulse.
The solid line represents MDIQKD with perfect trustworthy BB84 senders' devices,  MU without optics misalignment, which can distinguish one Bell state $|\phi^+\rangle_{CD}$, and no Eve's attack;
the dashed line is for QMDIQKD with perfect BB84 senders' devices (but not trusted by Alice and Bob), MU without optics misalignment, which can distinguish two Bell states $|\phi^+\rangle_{CD}$ and $|\psi^+\rangle_{CD}$, and no Eve's attack.
\label{fig-1}}
\end{figure}

 In Fig.~\ref{fig-1}, the secure-key rate for MDIQKD is given by the solid line, in which we assume that Alice and Bob are aware that their encoding states are perfect, MU can only distinguish Bell state $|\phi^+\rangle_{CD}$, and Eve is passive.
The secure-key rate for QMDIQKD is given by the dashed line, in which
we have ideal BB84 senders (but we do not trust them now), MU can distinguish two Bell states $|\phi^+\rangle_{CD}$ and $|\psi^+\rangle_{CD}$, and Eve is passive.

From Fig.~\ref{fig-1}, we see that QMDIQKD with only the two-dimensional assumption can offer near 160-km QKD service between Alice and Bob (i.e., twice the distance between the MU and Alice or Bob). One should note that the two lines in Fig.~\ref{fig-1} do not converge at $L=0$ due to the errors introduced by dark counts.
Here, the detection efficiency is $10\%$ and thus dark counts occur.
It can be seen that these errors lower the key rate more significantly in QMDIQKD than in the original MDIQKD.  On the other hand, if the efficiency is $100\%$, thus eliminating the effect of dark counts, the two lines will converge at $L=0$. To see this, if $L=0$ and the efficiency of SPDs is $1$, we must have $p(1|0,1)=p(1|3,2)=p(2|0,0)=p(2|1,1)=p(2|3,2)=0$ and $p(1|00)=p(1|11)=1$.
Then, $\xi$=0 and $\zeta$=0 according to their definitions in \eqref{eqn-constraints2}. By \eqref{eqn-constraints4}, we have $C_{30} C'_{20} = C_{31} C'_{21}$, since $p(1|0,0)=p(1|1,1)=1$. Also, note that $C_{30} C'_{20}\neq 0$ in this case. Then, by \eqref{eqn-max-exp}, we obtain $\epsilon = 0$, $e_p =0$, and finally $R=1$. And in this ideal situation, the original MDIQKD will also have $R=1$. Thus, the lines will converge in this case.

\section{Conclusion} \label{sec-conclusion}
In this paper, we have proved the security of a new MDIQKD scheme where uncharacterized qubit source systems are used. The difference between MDIQKD and QMDIQKD lies in the knowledge of Alice and Bob on their encoding states and the requirement of MU announcement. In MDIQKD, the encoded states are assumed to be well characterized, and the MU only needs to identify any one of the Bell states \cite{Lo:MDIQKD:2012}. In QMDIQKD, Alice and Bob do not need to worry about their encoding systems except that they need to be sure about the source-state dimensions, and the MU can distinguish at least two Bell states. The simulation results show that this scheme is practical. Thus, our work advances MDIQKD to be more device independent while keeping its main advantage: --- high loss and error tolerance.

There are a few extensions to our work that can be done in the future.
\begin{itemize}
\item
Our security proof assumes that Eve's attack is collective. This restriction can be removed by employing the recently developed security proofs \cite{Caves:deFinetti:2002,Renner:deFinetti:2009} to make our protocol secure against the most general attacks. It is an interesting future project to extend our proof to the most general attack case in the finite-size key scenario.

\item
Currently, the QMDIQKD protocol is restricted to 2-dimensional (qubit) source systems.
However, direct extension of this protocol to higher dimensions is not fruitful.  When the dimension is four or above, the protocol becomes insecure, since the four BB84 states can be unambiguously represented by a four-dimensional state.

\item
Weak coherent sources are widely used in QKD experiments. The decoy-state method can be employed to solve the multi-photon state issue. To extend the QMDIQKD protocol to coherent state sources, we can combine this protocol with the decoy-state method. Since the decoy-state method is proven to be secure under the Gottesman, Lo, Lutkenhaus, and Preskill (GLLP) scenario \cite{GLLP:2004,Lo:Decoy:2005}, such an extension is expected to be natural.
\end{itemize}

\section*{Acknowledgments}
The authors thank O.~Gittsovich, H.-K.~Lo, N.~L\"utkenhaus, B.~Qi, K.~Tamaki, and F.~Xu for helpful discussions. This work was supported by the National Basic Research Program of China (Grants No. 2011CBA00200 and No. 2011CB921200), National Natural Science Foundation of China (Grants No. 61101137 and No. 61201239). X.~M.~gratefully acknowledges the financial support from the National Basic Research Program of China Grants No.~2011CBA00300 and No.~2011CBA00301; and the 1000 Youth Fellowship program in China. C.-H.~F.~F.~gratefully acknowledges the financial support of RGC Grant No.~700712P from the HKSAR Government.

%
%
%

\bibliographystyle{apsrev4-1}

\bibliography{Biblisource}

\begin{thebibliography}{47}%
\makeatletter
\providecommand \@ifxundefined [1]{%
 \@ifx{#1\undefined}
}%
\providecommand \@ifnum [1]{%
 \ifnum #1\expandafter \@firstoftwo
 \else \expandafter \@secondoftwo
 \fi
}%
\providecommand \@ifx [1]{%
 \ifx #1\expandafter \@firstoftwo
 \else \expandafter \@secondoftwo
 \fi
}%
\providecommand \natexlab [1]{#1}%
\providecommand \enquote  [1]{``#1''}%
\providecommand \bibnamefont  [1]{#1}%
\providecommand \bibfnamefont [1]{#1}%
\providecommand \citenamefont [1]{#1}%
\providecommand \href@noop [0]{\@secondoftwo}%
\providecommand \href [0]{\begingroup \@sanitize@url \@href}%
\providecommand \@href[1]{\@@startlink{#1}\@@href}%
\providecommand \@@href[1]{\endgroup#1\@@endlink}%
\providecommand \@sanitize@url [0]{\catcode `\\12\catcode `\$12\catcode
  `\&12\catcode `\#12\catcode `\^12\catcode `\_12\catcode `\%12\relax}%
\providecommand \@@startlink[1]{}%
\providecommand \@@endlink[0]{}%
\providecommand \url  [0]{\begingroup\@sanitize@url \@url }%
\providecommand \@url [1]{\endgroup\@href {#1}{\urlprefix }}%
\providecommand \urlprefix  [0]{URL }%
\providecommand \Eprint [0]{\href }%
\providecommand \doibase [0]{http://dx.doi.org/}%
\providecommand \selectlanguage [0]{\@gobble}%
\providecommand \bibinfo  [0]{\@secondoftwo}%
\providecommand \bibfield  [0]{\@secondoftwo}%
\providecommand \translation [1]{[#1]}%
\providecommand \BibitemOpen [0]{}%
\providecommand \bibitemStop [0]{}%
\providecommand \bibitemNoStop [0]{.\EOS\space}%
\providecommand \EOS [0]{\spacefactor3000\relax}%
\providecommand \BibitemShut  [1]{\csname bibitem#1\endcsname}%
\let\auto@bib@innerbib\@empty
\bibitem [{\citenamefont {Bennett}\ and\ \citenamefont
  {Brassard}(1984)}]{Bennett:BB84:1984}%
  \BibitemOpen
  \bibfield  {author} {\bibinfo {author} {\bibfnamefont {C.~H.}\ \bibnamefont
  {Bennett}}\ and\ \bibinfo {author} {\bibfnamefont {G.}~\bibnamefont
  {Brassard}},\ }in\ \href@noop {} {\emph {\bibinfo {booktitle} {Proceedings of
  the IEEE International Conference on Computers, Systems and Signal
  Processing}}}\ (\bibinfo  {publisher} {IEEE Press},\ \bibinfo {address} {New
  York},\ \bibinfo {year} {1984})\ pp.\ \bibinfo {pages} {175--179}\BibitemShut
  {NoStop}%
\bibitem [{\citenamefont {Ekert}(1991)}]{Ekert:QKD:1991}%
  \BibitemOpen
  \bibfield  {author} {\bibinfo {author} {\bibfnamefont {A.~K.}\ \bibnamefont
  {Ekert}},\ }\href {\doibase 10.1103/PhysRevLett.67.661} {\bibfield  {journal}
  {\bibinfo  {journal} {Phys. Rev. Lett.}\ }\textbf {\bibinfo {volume} {67}},\
  \bibinfo {pages} {661} (\bibinfo {year} {1991})}\BibitemShut {NoStop}%
\bibitem [{\citenamefont {Mayers}(2001)}]{Mayers:Security:2001}%
  \BibitemOpen
  \bibfield  {author} {\bibinfo {author} {\bibfnamefont {D.}~\bibnamefont
  {Mayers}},\ }\href@noop {} {\bibfield  {journal} {\bibinfo  {journal}
  {Journal of the ACM (JACM)}\ }\textbf {\bibinfo {volume} {48}},\ \bibinfo
  {pages} {351} (\bibinfo {year} {2001})}\BibitemShut {NoStop}%
\bibitem [{\citenamefont {Lo}\ and\ \citenamefont
  {Chau}(1999)}]{Lo:QKDSecurity:1999}%
  \BibitemOpen
  \bibfield  {author} {\bibinfo {author} {\bibfnamefont {H.-K.}\ \bibnamefont
  {Lo}}\ and\ \bibinfo {author} {\bibfnamefont {H.~F.}\ \bibnamefont {Chau}},\
  }\href@noop {} {\bibfield  {journal} {\bibinfo  {journal} {Science}\ }\textbf
  {\bibinfo {volume} {283}},\ \bibinfo {pages} {2050} (\bibinfo {year}
  {1999})}\BibitemShut {NoStop}%
\bibitem [{\citenamefont {Shor}\ and\ \citenamefont
  {Preskill}(2000)}]{Shor:Preskill:2000}%
  \BibitemOpen
  \bibfield  {author} {\bibinfo {author} {\bibfnamefont {P.~W.}\ \bibnamefont
  {Shor}}\ and\ \bibinfo {author} {\bibfnamefont {J.}~\bibnamefont
  {Preskill}},\ }\href@noop {} {\bibfield  {journal} {\bibinfo  {journal}
  {Phys.~Rev.~Lett.~}\ }\textbf {\bibinfo {volume} {85}},\ \bibinfo {pages}
  {441} (\bibinfo {year} {2000})}\BibitemShut {NoStop}%
\bibitem [{\citenamefont {Tamaki}\ \emph {et~al.}(2003)\citenamefont {Tamaki},
  \citenamefont {Koashi},\ and\ \citenamefont {Imoto}}]{Tamaki:B92:2003}%
  \BibitemOpen
  \bibfield  {author} {\bibinfo {author} {\bibfnamefont {K.}~\bibnamefont
  {Tamaki}}, \bibinfo {author} {\bibfnamefont {M.}~\bibnamefont {Koashi}}, \
  and\ \bibinfo {author} {\bibfnamefont {N.}~\bibnamefont {Imoto}},\ }\href
  {\doibase 10.1103/PhysRevLett.90.167904} {\bibfield  {journal} {\bibinfo
  {journal} {Phys. Rev. Lett.}\ }\textbf {\bibinfo {volume} {90}},\ \bibinfo
  {pages} {167904} (\bibinfo {year} {2003})}\BibitemShut {NoStop}%
\bibitem [{\citenamefont {Boileau}\ \emph {et~al.}(2005)\citenamefont
  {Boileau}, \citenamefont {Tamaki}, \citenamefont {Batuwantudawe},
  \citenamefont {Laflamme},\ and\ \citenamefont {Renes}}]{Boileau:3state:2005}%
  \BibitemOpen
  \bibfield  {author} {\bibinfo {author} {\bibfnamefont {J.-C.}\ \bibnamefont
  {Boileau}}, \bibinfo {author} {\bibfnamefont {K.}~\bibnamefont {Tamaki}},
  \bibinfo {author} {\bibfnamefont {J.}~\bibnamefont {Batuwantudawe}}, \bibinfo
  {author} {\bibfnamefont {R.}~\bibnamefont {Laflamme}}, \ and\ \bibinfo
  {author} {\bibfnamefont {J.~M.}\ \bibnamefont {Renes}},\ }\href {\doibase
  10.1103/PhysRevLett.94.040503} {\bibfield  {journal} {\bibinfo  {journal}
  {Phys. Rev. Lett.}\ }\textbf {\bibinfo {volume} {94}},\ \bibinfo {pages}
  {040503} (\bibinfo {year} {2005})}\BibitemShut {NoStop}%
\bibitem [{\citenamefont {Yin}\ \emph {et~al.}(2010)\citenamefont {Yin},
  \citenamefont {Li}, \citenamefont {Chen}, \citenamefont {Han},\ and\
  \citenamefont {Guo}}]{PhysRevA.82.042335}%
  \BibitemOpen
  \bibfield  {author} {\bibinfo {author} {\bibfnamefont {Z.-Q.}\ \bibnamefont
  {Yin}}, \bibinfo {author} {\bibfnamefont {H.-W.}\ \bibnamefont {Li}},
  \bibinfo {author} {\bibfnamefont {W.}~\bibnamefont {Chen}}, \bibinfo {author}
  {\bibfnamefont {Z.-F.}\ \bibnamefont {Han}}, \ and\ \bibinfo {author}
  {\bibfnamefont {G.-C.}\ \bibnamefont {Guo}},\ }\href {\doibase
  10.1103/PhysRevA.82.042335} {\bibfield  {journal} {\bibinfo  {journal} {Phys.
  Rev. A}\ }\textbf {\bibinfo {volume} {82}},\ \bibinfo {pages} {042335}
  (\bibinfo {year} {2010})}\BibitemShut {NoStop}%
\bibitem [{\citenamefont {Zhao}\ \emph
  {et~al.}(2006{\natexlab{a}})\citenamefont {Zhao}, \citenamefont {Qi},
  \citenamefont {Ma}, \citenamefont {Lo},\ and\ \citenamefont
  {Qian}}]{Zhao:DecoyExp:2006}%
  \BibitemOpen
  \bibfield  {author} {\bibinfo {author} {\bibfnamefont {Y.}~\bibnamefont
  {Zhao}}, \bibinfo {author} {\bibfnamefont {B.}~\bibnamefont {Qi}}, \bibinfo
  {author} {\bibfnamefont {X.}~\bibnamefont {Ma}}, \bibinfo {author}
  {\bibfnamefont {H.-K.}\ \bibnamefont {Lo}}, \ and\ \bibinfo {author}
  {\bibfnamefont {L.}~\bibnamefont {Qian}},\ }\href@noop {} {\bibfield
  {journal} {\bibinfo  {journal} {Phys.~Rev.~Lett.~}\ }\textbf {\bibinfo
  {volume} {96}},\ \bibinfo {pages} {070502} (\bibinfo {year}
  {2006}{\natexlab{a}})}\BibitemShut {NoStop}%
\bibitem [{\citenamefont {Rosenberg}\ \emph {et~al.}(2007)\citenamefont
  {Rosenberg}, \citenamefont {Harrington}, \citenamefont {Rice}, \citenamefont
  {Hiskett}, \citenamefont {Peterson}, \citenamefont {Hughes}, \citenamefont
  {Lita}, \citenamefont {Nam},\ and\ \citenamefont
  {Nordholt}}]{Rosenberg:ExpDecoy:2007}%
  \BibitemOpen
  \bibfield  {author} {\bibinfo {author} {\bibfnamefont {D.}~\bibnamefont
  {Rosenberg}}, \bibinfo {author} {\bibfnamefont {J.~W.}\ \bibnamefont
  {Harrington}}, \bibinfo {author} {\bibfnamefont {P.~R.}\ \bibnamefont
  {Rice}}, \bibinfo {author} {\bibfnamefont {P.~A.}\ \bibnamefont {Hiskett}},
  \bibinfo {author} {\bibfnamefont {C.~G.}\ \bibnamefont {Peterson}}, \bibinfo
  {author} {\bibfnamefont {R.~J.}\ \bibnamefont {Hughes}}, \bibinfo {author}
  {\bibfnamefont {A.~E.}\ \bibnamefont {Lita}}, \bibinfo {author}
  {\bibfnamefont {S.~W.}\ \bibnamefont {Nam}}, \ and\ \bibinfo {author}
  {\bibfnamefont {J.~E.}\ \bibnamefont {Nordholt}},\ }\href {\doibase
  10.1103/PhysRevLett.98.010503} {\bibfield  {journal} {\bibinfo  {journal}
  {Phys. Rev. Lett.}\ }\textbf {\bibinfo {volume} {98}},\ \bibinfo {pages}
  {010503} (\bibinfo {year} {2007})}\BibitemShut {NoStop}%
\bibitem [{\citenamefont {Schmitt-Manderbach}\ \emph
  {et~al.}(2007)\citenamefont {Schmitt-Manderbach}, \citenamefont {Weier},
  \citenamefont {F\"urst}, \citenamefont {Ursin}, \citenamefont {Tiefenbacher},
  \citenamefont {Scheidl}, \citenamefont {Perdigues}, \citenamefont {Sodnik},
  \citenamefont {Kurtsiefer}, \citenamefont {Rarity}, \citenamefont
  {Zeilinger},\ and\ \citenamefont {Weinfurter}}]{Zeilinger:ExpDecoy:2007}%
  \BibitemOpen
  \bibfield  {author} {\bibinfo {author} {\bibfnamefont {T.}~\bibnamefont
  {Schmitt-Manderbach}}, \bibinfo {author} {\bibfnamefont {H.}~\bibnamefont
  {Weier}}, \bibinfo {author} {\bibfnamefont {M.}~\bibnamefont {F\"urst}},
  \bibinfo {author} {\bibfnamefont {R.}~\bibnamefont {Ursin}}, \bibinfo
  {author} {\bibfnamefont {F.}~\bibnamefont {Tiefenbacher}}, \bibinfo {author}
  {\bibfnamefont {T.}~\bibnamefont {Scheidl}}, \bibinfo {author} {\bibfnamefont
  {J.}~\bibnamefont {Perdigues}}, \bibinfo {author} {\bibfnamefont
  {Z.}~\bibnamefont {Sodnik}}, \bibinfo {author} {\bibfnamefont
  {C.}~\bibnamefont {Kurtsiefer}}, \bibinfo {author} {\bibfnamefont {J.~G.}\
  \bibnamefont {Rarity}}, \bibinfo {author} {\bibfnamefont {A.}~\bibnamefont
  {Zeilinger}}, \ and\ \bibinfo {author} {\bibfnamefont {H.}~\bibnamefont
  {Weinfurter}},\ }\href {\doibase 10.1103/PhysRevLett.98.010504} {\bibfield
  {journal} {\bibinfo  {journal} {Phys. Rev. Lett.}\ }\textbf {\bibinfo
  {volume} {98}},\ \bibinfo {pages} {010504} (\bibinfo {year}
  {2007})}\BibitemShut {NoStop}%
\bibitem [{\citenamefont {Peng}\ \emph {et~al.}(2007)\citenamefont {Peng},
  \citenamefont {Zhang}, \citenamefont {Yang}, \citenamefont {Gao},
  \citenamefont {Ma}, \citenamefont {Yin}, \citenamefont {Zeng}, \citenamefont
  {Yang}, \citenamefont {Wang},\ and\ \citenamefont
  {Pan}}]{Peng:ExpDecoy:2007}%
  \BibitemOpen
  \bibfield  {author} {\bibinfo {author} {\bibfnamefont {C.-Z.}\ \bibnamefont
  {Peng}}, \bibinfo {author} {\bibfnamefont {J.}~\bibnamefont {Zhang}},
  \bibinfo {author} {\bibfnamefont {D.}~\bibnamefont {Yang}}, \bibinfo {author}
  {\bibfnamefont {W.-B.}\ \bibnamefont {Gao}}, \bibinfo {author} {\bibfnamefont
  {H.-X.}\ \bibnamefont {Ma}}, \bibinfo {author} {\bibfnamefont
  {H.}~\bibnamefont {Yin}}, \bibinfo {author} {\bibfnamefont {H.-P.}\
  \bibnamefont {Zeng}}, \bibinfo {author} {\bibfnamefont {T.}~\bibnamefont
  {Yang}}, \bibinfo {author} {\bibfnamefont {X.-B.}\ \bibnamefont {Wang}}, \
  and\ \bibinfo {author} {\bibfnamefont {J.-W.}\ \bibnamefont {Pan}},\ }\href
  {\doibase 10.1103/PhysRevLett.98.010505} {\bibfield  {journal} {\bibinfo
  {journal} {Phys. Rev. Lett.}\ }\textbf {\bibinfo {volume} {98}},\ \bibinfo
  {pages} {010505} (\bibinfo {year} {2007})}\BibitemShut {NoStop}%
\bibitem [{\citenamefont {Zhao}\ \emph
  {et~al.}(2006{\natexlab{b}})\citenamefont {Zhao}, \citenamefont {Qi},
  \citenamefont {Ma}, \citenamefont {Lo},\ and\ \citenamefont
  {Qian}}]{Zhao:Decoy60km:2006}%
  \BibitemOpen
  \bibfield  {author} {\bibinfo {author} {\bibfnamefont {Y.}~\bibnamefont
  {Zhao}}, \bibinfo {author} {\bibfnamefont {B.}~\bibnamefont {Qi}}, \bibinfo
  {author} {\bibfnamefont {X.}~\bibnamefont {Ma}}, \bibinfo {author}
  {\bibfnamefont {H.-K.}\ \bibnamefont {Lo}}, \ and\ \bibinfo {author}
  {\bibfnamefont {L.}~\bibnamefont {Qian}},\ }in\ \href@noop {} {\emph
  {\bibinfo {booktitle} {Proc.~of IEEE ISIT}}}\ (\bibinfo  {publisher} {IEEE},\
  \bibinfo {year} {2006})\ p.\ \bibinfo {pages} {2094}\BibitemShut {NoStop}%
\bibitem [{\citenamefont {Yuan}\ \emph {et~al.}(2007)\citenamefont {Yuan},
  \citenamefont {Sharpe},\ and\ \citenamefont {Shields}}]{Yuan:ExpDecoy2007}%
  \BibitemOpen
  \bibfield  {author} {\bibinfo {author} {\bibfnamefont {Z.~L.}\ \bibnamefont
  {Yuan}}, \bibinfo {author} {\bibfnamefont {A.~W.}\ \bibnamefont {Sharpe}}, \
  and\ \bibinfo {author} {\bibfnamefont {A.~J.}\ \bibnamefont {Shields}},\
  }\href@noop {} {\bibfield  {journal} {\bibinfo  {journal}
  {Appl.~Phys.~Lett.}\ }\textbf {\bibinfo {volume} {90}},\ \bibinfo {pages}
  {011118} (\bibinfo {year} {2007})}\BibitemShut {NoStop}%
\bibitem [{\citenamefont {Takesue}\ \emph {et~al.}(2007)\citenamefont
  {Takesue}, \citenamefont {Nam}, \citenamefont {Zhang}, \citenamefont
  {Hadfield}, \citenamefont {Honjo}, \citenamefont {Tamaki},\ and\
  \citenamefont {Yamamoto}}]{Takesue:40dBQKD:2007}%
  \BibitemOpen
  \bibfield  {author} {\bibinfo {author} {\bibfnamefont {H.}~\bibnamefont
  {Takesue}}, \bibinfo {author} {\bibfnamefont {S.}~\bibnamefont {Nam}},
  \bibinfo {author} {\bibfnamefont {Q.}~\bibnamefont {Zhang}}, \bibinfo
  {author} {\bibfnamefont {R.}~\bibnamefont {Hadfield}}, \bibinfo {author}
  {\bibfnamefont {T.}~\bibnamefont {Honjo}}, \bibinfo {author} {\bibfnamefont
  {K.}~\bibnamefont {Tamaki}}, \ and\ \bibinfo {author} {\bibfnamefont
  {Y.}~\bibnamefont {Yamamoto}},\ }\href@noop {} {\bibfield  {journal}
  {\bibinfo  {journal} {Nature Photonics}\ }\textbf {\bibinfo {volume} {1}},\
  \bibinfo {pages} {343} (\bibinfo {year} {2007})}\BibitemShut {NoStop}%
\bibitem [{\citenamefont {Stucki}\ \emph {et~al.}(2002)\citenamefont {Stucki},
  \citenamefont {Gisin}, \citenamefont {Guinnard}, \citenamefont {Ribordy},\
  and\ \citenamefont {Zbinden}}]{Stucki:250kmQKD:2009}%
  \BibitemOpen
  \bibfield  {author} {\bibinfo {author} {\bibfnamefont {D.}~\bibnamefont
  {Stucki}}, \bibinfo {author} {\bibfnamefont {N.}~\bibnamefont {Gisin}},
  \bibinfo {author} {\bibfnamefont {O.}~\bibnamefont {Guinnard}}, \bibinfo
  {author} {\bibfnamefont {G.}~\bibnamefont {Ribordy}}, \ and\ \bibinfo
  {author} {\bibfnamefont {H.}~\bibnamefont {Zbinden}},\ }\href@noop {}
  {\bibfield  {journal} {\bibinfo  {journal} {New J. of Phys.}\ }\textbf
  {\bibinfo {volume} {4}},\ \bibinfo {pages} {41} (\bibinfo {year}
  {2002})}\BibitemShut {NoStop}%
\bibitem [{\citenamefont {Wang}\ \emph {et~al.}(2012)\citenamefont {Wang},
  \citenamefont {Chen}, \citenamefont {Guo}, \citenamefont {Yin}, \citenamefont
  {Li}, \citenamefont {Zhou}, \citenamefont {Guo},\ and\ \citenamefont
  {Han}}]{Wang:260kmQKD:2012}%
  \BibitemOpen
  \bibfield  {author} {\bibinfo {author} {\bibfnamefont {S.}~\bibnamefont
  {Wang}}, \bibinfo {author} {\bibfnamefont {W.}~\bibnamefont {Chen}}, \bibinfo
  {author} {\bibfnamefont {J.-F.}\ \bibnamefont {Guo}}, \bibinfo {author}
  {\bibfnamefont {Z.-Q.}\ \bibnamefont {Yin}}, \bibinfo {author} {\bibfnamefont
  {H.-W.}\ \bibnamefont {Li}}, \bibinfo {author} {\bibfnamefont
  {Z.}~\bibnamefont {Zhou}}, \bibinfo {author} {\bibfnamefont {G.-C.}\
  \bibnamefont {Guo}}, \ and\ \bibinfo {author} {\bibfnamefont {Z.-F.}\
  \bibnamefont {Han}},\ }\href {\doibase 10.1364/OL.37.001008} {\bibfield
  {journal} {\bibinfo  {journal} {Opt. Lett.}\ }\textbf {\bibinfo {volume}
  {37}},\ \bibinfo {pages} {1008} (\bibinfo {year} {2012})}\BibitemShut
  {NoStop}%
\bibitem [{\citenamefont {Frolich}\ \emph {et~al.}(2013)\citenamefont
  {Frolich}, \citenamefont {Dynes}, \citenamefont {Lucamarini}, \citenamefont
  {Sharpe}, \citenamefont {Yuan},\ and\ \citenamefont
  {Shields}}]{Frolich:QKDnet:2013}%
  \BibitemOpen
  \bibfield  {author} {\bibinfo {author} {\bibfnamefont {B.}~\bibnamefont
  {Frolich}}, \bibinfo {author} {\bibfnamefont {J.~F.}\ \bibnamefont {Dynes}},
  \bibinfo {author} {\bibfnamefont {M.}~\bibnamefont {Lucamarini}}, \bibinfo
  {author} {\bibfnamefont {A.~W.}\ \bibnamefont {Sharpe}}, \bibinfo {author}
  {\bibfnamefont {Z.}~\bibnamefont {Yuan}}, \ and\ \bibinfo {author}
  {\bibfnamefont {A.~J.}\ \bibnamefont {Shields}},\ }\href {\doibase
  10.1038/nature12493} {\bibfield  {journal} {\bibinfo  {journal} {Nature}\
  }\textbf {\bibinfo {volume} {501}},\ \bibinfo {pages} {69} (\bibinfo {year}
  {2013})}\BibitemShut {NoStop}%
\bibitem [{\citenamefont {Scarani}\ \emph {et~al.}(2009)\citenamefont
  {Scarani}, \citenamefont {Bechmann-Pasquinucci}, \citenamefont {Cerf},
  \citenamefont {Du\ifmmode~\check{s}\else \v{s}\fi{}ek}, \citenamefont
  {L\"utkenhaus},\ and\ \citenamefont {Peev}}]{Scarani:QKDrev:2008}%
  \BibitemOpen
  \bibfield  {author} {\bibinfo {author} {\bibfnamefont {V.}~\bibnamefont
  {Scarani}}, \bibinfo {author} {\bibfnamefont {H.}~\bibnamefont
  {Bechmann-Pasquinucci}}, \bibinfo {author} {\bibfnamefont {N.~J.}\
  \bibnamefont {Cerf}}, \bibinfo {author} {\bibfnamefont {M.}~\bibnamefont
  {Du\ifmmode~\check{s}\else \v{s}\fi{}ek}}, \bibinfo {author} {\bibfnamefont
  {N.}~\bibnamefont {L\"utkenhaus}}, \ and\ \bibinfo {author} {\bibfnamefont
  {M.}~\bibnamefont {Peev}},\ }\href {\doibase 10.1103/RevModPhys.81.1301}
  {\bibfield  {journal} {\bibinfo  {journal} {Rev. Mod. Phys.}\ }\textbf
  {\bibinfo {volume} {81}},\ \bibinfo {pages} {1301} (\bibinfo {year}
  {2009})}\BibitemShut {NoStop}%
\bibitem [{\citenamefont {Makarov}\ \emph {et~al.}(2006)\citenamefont
  {Makarov}, \citenamefont {Anisimov},\ and\ \citenamefont
  {Skaar}}]{MAS_Eff_06}%
  \BibitemOpen
  \bibfield  {author} {\bibinfo {author} {\bibfnamefont {V.}~\bibnamefont
  {Makarov}}, \bibinfo {author} {\bibfnamefont {A.}~\bibnamefont {Anisimov}}, \
  and\ \bibinfo {author} {\bibfnamefont {J.}~\bibnamefont {Skaar}},\
  }\href@noop {} {\bibfield  {journal} {\bibinfo  {journal} {Phys. Rev. A}\
  }\textbf {\bibinfo {volume} {74}},\ \bibinfo {pages} {022313} (\bibinfo
  {year} {2006})}\BibitemShut {NoStop}%
\bibitem [{\citenamefont {Makarov}\ and\ \citenamefont
  {Skaar}(2008)}]{Makarov:Fake:08}%
  \BibitemOpen
  \bibfield  {author} {\bibinfo {author} {\bibfnamefont {V.}~\bibnamefont
  {Makarov}}\ and\ \bibinfo {author} {\bibfnamefont {J.}~\bibnamefont
  {Skaar}},\ }\href@noop {} {\bibfield  {journal} {\bibinfo  {journal} {Quantum
  Inf.~Comput.~}\ }\textbf {\bibinfo {volume} {8}},\ \bibinfo {pages} {0622}
  (\bibinfo {year} {2008})}\BibitemShut {NoStop}%
\bibitem [{\citenamefont {Qi}\ \emph {et~al.}(2007)\citenamefont {Qi},
  \citenamefont {Fung}, \citenamefont {Lo},\ and\ \citenamefont
  {Ma}}]{Qi:TimeShift:2007}%
  \BibitemOpen
  \bibfield  {author} {\bibinfo {author} {\bibfnamefont {B.}~\bibnamefont
  {Qi}}, \bibinfo {author} {\bibfnamefont {C.-H.~F.}\ \bibnamefont {Fung}},
  \bibinfo {author} {\bibfnamefont {H.-K.}\ \bibnamefont {Lo}}, \ and\ \bibinfo
  {author} {\bibfnamefont {X.}~\bibnamefont {Ma}},\ }\href@noop {} {\bibfield
  {journal} {\bibinfo  {journal} {Quantum Inf.~Comput.}\ }\textbf {\bibinfo
  {volume} {7}},\ \bibinfo {pages} {073} (\bibinfo {year} {2007})}\BibitemShut
  {NoStop}%
\bibitem [{\citenamefont {Zhao}\ \emph {et~al.}(2008)\citenamefont {Zhao},
  \citenamefont {Fung}, \citenamefont {Qi}, \citenamefont {Chen},\ and\
  \citenamefont {Lo}}]{Zhao:TimeshiftExp:2008}%
  \BibitemOpen
  \bibfield  {author} {\bibinfo {author} {\bibfnamefont {Y.}~\bibnamefont
  {Zhao}}, \bibinfo {author} {\bibfnamefont {C.-H.~F.}\ \bibnamefont {Fung}},
  \bibinfo {author} {\bibfnamefont {B.}~\bibnamefont {Qi}}, \bibinfo {author}
  {\bibfnamefont {C.}~\bibnamefont {Chen}}, \ and\ \bibinfo {author}
  {\bibfnamefont {H.-K.}\ \bibnamefont {Lo}},\ }\href {\doibase
  10.1103/PhysRevA.78.042333} {\bibfield  {journal} {\bibinfo  {journal} {Phys.
  Rev. A}\ }\textbf {\bibinfo {volume} {78}},\ \bibinfo {pages} {042333}
  (\bibinfo {year} {2008})}\BibitemShut {NoStop}%
\bibitem [{\citenamefont {Fung}\ \emph {et~al.}(2007)\citenamefont {Fung},
  \citenamefont {Qi}, \citenamefont {Tamaki},\ and\ \citenamefont
  {Lo}}]{Fung:Remap:07}%
  \BibitemOpen
  \bibfield  {author} {\bibinfo {author} {\bibfnamefont {C.-H.~F.}\
  \bibnamefont {Fung}}, \bibinfo {author} {\bibfnamefont {B.}~\bibnamefont
  {Qi}}, \bibinfo {author} {\bibfnamefont {K.}~\bibnamefont {Tamaki}}, \ and\
  \bibinfo {author} {\bibfnamefont {H.-K.}\ \bibnamefont {Lo}},\ }\href
  {\doibase 10.1103/PhysRevA.75.032314} {\bibfield  {journal} {\bibinfo
  {journal} {Phys. Rev. A}\ }\textbf {\bibinfo {volume} {75}},\ \bibinfo
  {pages} {032314} (\bibinfo {year} {2007})}\BibitemShut {NoStop}%
\bibitem [{\citenamefont {Xu}\ \emph {et~al.}(2010)\citenamefont {Xu},
  \citenamefont {Qi},\ and\ \citenamefont {Lo}}]{Xu:PhaseRemap:2010}%
  \BibitemOpen
  \bibfield  {author} {\bibinfo {author} {\bibfnamefont {F.}~\bibnamefont
  {Xu}}, \bibinfo {author} {\bibfnamefont {B.}~\bibnamefont {Qi}}, \ and\
  \bibinfo {author} {\bibfnamefont {H.-K.}\ \bibnamefont {Lo}},\ }\href
  {http://stacks.iop.org/1367-2630/12/i=11/a=113026} {\bibfield  {journal}
  {\bibinfo  {journal} {New Journal of Physics}\ }\textbf {\bibinfo {volume}
  {12}},\ \bibinfo {pages} {113026} (\bibinfo {year} {2010})}\BibitemShut
  {NoStop}%
\bibitem [{\citenamefont {Lydersen}\ \emph {et~al.}(2010)\citenamefont
  {Lydersen}, \citenamefont {Wiechers}, \citenamefont {Wittmann}, \citenamefont
  {Elser}, \citenamefont {Skaar},\ and\ \citenamefont
  {Makarov}}]{Lydersen:Hacking:2010}%
  \BibitemOpen
  \bibfield  {author} {\bibinfo {author} {\bibfnamefont {L.}~\bibnamefont
  {Lydersen}}, \bibinfo {author} {\bibfnamefont {C.}~\bibnamefont {Wiechers}},
  \bibinfo {author} {\bibfnamefont {C.}~\bibnamefont {Wittmann}}, \bibinfo
  {author} {\bibfnamefont {D.}~\bibnamefont {Elser}}, \bibinfo {author}
  {\bibfnamefont {J.}~\bibnamefont {Skaar}}, \ and\ \bibinfo {author}
  {\bibfnamefont {V.}~\bibnamefont {Makarov}},\ }\href@noop {} {\bibfield
  {journal} {\bibinfo  {journal} {Nature photonics}\ }\textbf {\bibinfo
  {volume} {4}},\ \bibinfo {pages} {686} (\bibinfo {year} {2010})}\BibitemShut
  {NoStop}%
\bibitem [{\citenamefont {Gerhardt}\ \emph {et~al.}(2011)\citenamefont
  {Gerhardt}, \citenamefont {Liu}, \citenamefont {Lamas-Linares}, \citenamefont
  {Skaar}, \citenamefont {Kurtsiefer},\ and\ \citenamefont
  {Makarov}}]{Gerhardt:Blind:2011}%
  \BibitemOpen
  \bibfield  {author} {\bibinfo {author} {\bibfnamefont {I.}~\bibnamefont
  {Gerhardt}}, \bibinfo {author} {\bibfnamefont {Q.}~\bibnamefont {Liu}},
  \bibinfo {author} {\bibfnamefont {A.}~\bibnamefont {Lamas-Linares}}, \bibinfo
  {author} {\bibfnamefont {J.}~\bibnamefont {Skaar}}, \bibinfo {author}
  {\bibfnamefont {C.}~\bibnamefont {Kurtsiefer}}, \ and\ \bibinfo {author}
  {\bibfnamefont {V.}~\bibnamefont {Makarov}},\ }\href@noop {} {\bibfield
  {journal} {\bibinfo  {journal} {{Nature Communications}}\ }\textbf {\bibinfo
  {volume} {{2}}},\ \bibinfo {pages} {349} (\bibinfo {year}
  {{2011}})}\BibitemShut {NoStop}%
\bibitem [{\citenamefont {Tang}\ \emph
  {et~al.}(2013{\natexlab{a}})\citenamefont {Tang}, \citenamefont {Yin},
  \citenamefont {Ma}, \citenamefont {Fung}, \citenamefont {Liu}, \citenamefont
  {Yong}, \citenamefont {Chen}, \citenamefont {Peng}, \citenamefont {Chen},\
  and\ \citenamefont {Pan}}]{Tang:USD:2013}%
  \BibitemOpen
  \bibfield  {author} {\bibinfo {author} {\bibfnamefont {Y.-L.}\ \bibnamefont
  {Tang}}, \bibinfo {author} {\bibfnamefont {H.-L.}\ \bibnamefont {Yin}},
  \bibinfo {author} {\bibfnamefont {X.}~\bibnamefont {Ma}}, \bibinfo {author}
  {\bibfnamefont {C.-H.~F.}\ \bibnamefont {Fung}}, \bibinfo {author}
  {\bibfnamefont {Y.}~\bibnamefont {Liu}}, \bibinfo {author} {\bibfnamefont
  {H.-L.}\ \bibnamefont {Yong}}, \bibinfo {author} {\bibfnamefont {T.-Y.}\
  \bibnamefont {Chen}}, \bibinfo {author} {\bibfnamefont {C.-Z.}\ \bibnamefont
  {Peng}}, \bibinfo {author} {\bibfnamefont {Z.-B.}\ \bibnamefont {Chen}}, \
  and\ \bibinfo {author} {\bibfnamefont {J.-W.}\ \bibnamefont {Pan}},\ }\href
  {\doibase 10.1103/PhysRevA.88.022308} {\bibfield  {journal} {\bibinfo
  {journal} {Phys. Rev. A}\ }\textbf {\bibinfo {volume} {88}},\ \bibinfo
  {pages} {022308} (\bibinfo {year} {2013}{\natexlab{a}})}\BibitemShut
  {NoStop}%
\bibitem [{\citenamefont {Mayers}\ and\ \citenamefont
  {Yao}(1998)}]{MayersYao_98}%
  \BibitemOpen
  \bibfield  {author} {\bibinfo {author} {\bibfnamefont {D.}~\bibnamefont
  {Mayers}}\ and\ \bibinfo {author} {\bibfnamefont {A.}~\bibnamefont {Yao}},\
  }in\ \href@noop {} {\emph {\bibinfo {booktitle} {FOCS, 39th Annual Symposium
  on Foundations of Computer Science}}}\ (\bibinfo  {publisher} {IEEE, Computer
  Society Press, Los Alamitos},\ \bibinfo {year} {1998})\ p.\ \bibinfo {pages}
  {503}\BibitemShut {NoStop}%
\bibitem [{\citenamefont {Ac\'{\i}n}\ \emph {et~al.}(2007)\citenamefont
  {Ac\'{\i}n}, \citenamefont {Brunner}, \citenamefont {Gisin}, \citenamefont
  {Massar}, \citenamefont {Pironio},\ and\ \citenamefont
  {Scarani}}]{Acin:DeviceIn:07}%
  \BibitemOpen
  \bibfield  {author} {\bibinfo {author} {\bibfnamefont {A.}~\bibnamefont
  {Ac\'{\i}n}}, \bibinfo {author} {\bibfnamefont {N.}~\bibnamefont {Brunner}},
  \bibinfo {author} {\bibfnamefont {N.}~\bibnamefont {Gisin}}, \bibinfo
  {author} {\bibfnamefont {S.}~\bibnamefont {Massar}}, \bibinfo {author}
  {\bibfnamefont {S.}~\bibnamefont {Pironio}}, \ and\ \bibinfo {author}
  {\bibfnamefont {V.}~\bibnamefont {Scarani}},\ }\href {\doibase
  10.1103/PhysRevLett.98.230501} {\bibfield  {journal} {\bibinfo  {journal}
  {Physical Review Letters}\ }\textbf {\bibinfo {volume} {98}},\ \bibinfo {eid}
  {230501} (\bibinfo {year} {2007})}\BibitemShut {NoStop}%
\bibitem [{\citenamefont {Lim}\ \emph {et~al.}(2013)\citenamefont {Lim},
  \citenamefont {Portmann}, \citenamefont {Tomamichel}, \citenamefont
  {Renner},\ and\ \citenamefont {Gisin}}]{Lim:DIQKD:2013}%
  \BibitemOpen
  \bibfield  {author} {\bibinfo {author} {\bibfnamefont {C.~C.~W.}\
  \bibnamefont {Lim}}, \bibinfo {author} {\bibfnamefont {C.}~\bibnamefont
  {Portmann}}, \bibinfo {author} {\bibfnamefont {M.}~\bibnamefont
  {Tomamichel}}, \bibinfo {author} {\bibfnamefont {R.}~\bibnamefont {Renner}},
  \ and\ \bibinfo {author} {\bibfnamefont {N.}~\bibnamefont {Gisin}},\ }\href
  {\doibase 10.1103/PhysRevX.3.031006} {\bibfield  {journal} {\bibinfo
  {journal} {Phys. Rev. X}\ }\textbf {\bibinfo {volume} {3}},\ \bibinfo {pages}
  {031006} (\bibinfo {year} {2013})}\BibitemShut {NoStop}%
\bibitem [{\citenamefont {Ma}\ and\ \citenamefont
  {L\"utkenhaus}(2012)}]{MXF:DDIQKD:2012}%
  \BibitemOpen
  \bibfield  {author} {\bibinfo {author} {\bibfnamefont {X.}~\bibnamefont
  {Ma}}\ and\ \bibinfo {author} {\bibfnamefont {N.}~\bibnamefont
  {L\"utkenhaus}},\ }\href@noop {} {\bibfield  {journal} {\bibinfo  {journal}
  {Quantum Inf.~Comput.}\ }\textbf {\bibinfo {volume} {12}},\ \bibinfo {pages}
  {0203} (\bibinfo {year} {2012})}\BibitemShut {NoStop}%
\bibitem [{\citenamefont {Paw\l{}owski}\ and\ \citenamefont
  {Brunner}(2011)}]{Pawlowski:Semi:2011}%
  \BibitemOpen
  \bibfield  {author} {\bibinfo {author} {\bibfnamefont {M.}~\bibnamefont
  {Paw\l{}owski}}\ and\ \bibinfo {author} {\bibfnamefont {N.}~\bibnamefont
  {Brunner}},\ }\href {\doibase 10.1103/PhysRevA.84.010302} {\bibfield
  {journal} {\bibinfo  {journal} {Phys. Rev. A}\ }\textbf {\bibinfo {volume}
  {84}},\ \bibinfo {pages} {010302} (\bibinfo {year} {2011})}\BibitemShut
  {NoStop}%
\bibitem [{\citenamefont {Lo}\ \emph {et~al.}(2012)\citenamefont {Lo},
  \citenamefont {Curty},\ and\ \citenamefont {Qi}}]{Lo:MDIQKD:2012}%
  \BibitemOpen
  \bibfield  {author} {\bibinfo {author} {\bibfnamefont {H.-K.}\ \bibnamefont
  {Lo}}, \bibinfo {author} {\bibfnamefont {M.}~\bibnamefont {Curty}}, \ and\
  \bibinfo {author} {\bibfnamefont {B.}~\bibnamefont {Qi}},\ }\href {\doibase
  10.1103/PhysRevLett.108.130503} {\bibfield  {journal} {\bibinfo  {journal}
  {Phys. Rev. Lett.}\ }\textbf {\bibinfo {volume} {108}},\ \bibinfo {pages}
  {130503} (\bibinfo {year} {2012})}\BibitemShut {NoStop}%
\bibitem [{\citenamefont {Braunstein}\ and\ \citenamefont
  {Pirandola}(2012)}]{Braunstein:MIQKD:2012}%
  \BibitemOpen
  \bibfield  {author} {\bibinfo {author} {\bibfnamefont {S.~L.}\ \bibnamefont
  {Braunstein}}\ and\ \bibinfo {author} {\bibfnamefont {S.}~\bibnamefont
  {Pirandola}},\ }\href {\doibase 10.1103/PhysRevLett.108.130502} {\bibfield
  {journal} {\bibinfo  {journal} {Phys. Rev. Lett.}\ }\textbf {\bibinfo
  {volume} {108}},\ \bibinfo {pages} {130502} (\bibinfo {year}
  {2012})}\BibitemShut {NoStop}%
\bibitem [{\citenamefont {Rubenok}\ \emph {et~al.}(2012)\citenamefont
  {Rubenok}, \citenamefont {Slater}, \citenamefont {Chan}, \citenamefont
  {Lucio-Martinez},\ and\ \citenamefont {Tittel}}]{Rubenok:MIQKDexp:2012}%
  \BibitemOpen
  \bibfield  {author} {\bibinfo {author} {\bibfnamefont {A.}~\bibnamefont
  {Rubenok}}, \bibinfo {author} {\bibfnamefont {J.}~\bibnamefont {Slater}},
  \bibinfo {author} {\bibfnamefont {P.}~\bibnamefont {Chan}}, \bibinfo {author}
  {\bibfnamefont {I.}~\bibnamefont {Lucio-Martinez}}, \ and\ \bibinfo {author}
  {\bibfnamefont {W.}~\bibnamefont {Tittel}},\ }\href@noop {} {\bibfield
  {journal} {\bibinfo  {journal} {Arxiv preprint arXiv:1204.0738}\ } (\bibinfo
  {year} {2012})}\BibitemShut {NoStop}%
\bibitem [{\citenamefont {da~Silva}\ \emph {et~al.}(2012)\citenamefont
  {da~Silva}, \citenamefont {Vitoreti}, \citenamefont {Xavier}, \citenamefont
  {Tempor{\~a}o},\ and\ \citenamefont {von~der Weid}}]{daSilva:MIQKD:2012}%
  \BibitemOpen
  \bibfield  {author} {\bibinfo {author} {\bibfnamefont {T.}~\bibnamefont
  {da~Silva}}, \bibinfo {author} {\bibfnamefont {D.}~\bibnamefont {Vitoreti}},
  \bibinfo {author} {\bibfnamefont {G.}~\bibnamefont {Xavier}}, \bibinfo
  {author} {\bibfnamefont {G.}~\bibnamefont {Tempor{\~a}o}}, \ and\ \bibinfo
  {author} {\bibfnamefont {J.}~\bibnamefont {von~der Weid}},\ }\href@noop {}
  {\bibfield  {journal} {\bibinfo  {journal} {Arxiv preprint arXiv:1207.6345}\
  } (\bibinfo {year} {2012})}\BibitemShut {NoStop}%
\bibitem [{\citenamefont {Liu}\ \emph {et~al.}(2012)\citenamefont {Liu},
  \citenamefont {Chen}, \citenamefont {Wang}, \citenamefont {Liang},
  \citenamefont {Shentu}, \citenamefont {Wang}, \citenamefont {Cui},
  \citenamefont {Yin}, \citenamefont {Liu}, \citenamefont {Li}, \citenamefont
  {Ma}, \citenamefont {Pelc}, \citenamefont {Fejer}, \citenamefont {Zhang},\
  and\ \citenamefont {Pan}}]{liu:MIQKDexp:2012}%
  \BibitemOpen
  \bibfield  {author} {\bibinfo {author} {\bibfnamefont {Y.}~\bibnamefont
  {Liu}}, \bibinfo {author} {\bibfnamefont {T.-Y.}\ \bibnamefont {Chen}},
  \bibinfo {author} {\bibfnamefont {L.-J.}\ \bibnamefont {Wang}}, \bibinfo
  {author} {\bibfnamefont {H.}~\bibnamefont {Liang}}, \bibinfo {author}
  {\bibfnamefont {G.-L.}\ \bibnamefont {Shentu}}, \bibinfo {author}
  {\bibfnamefont {J.}~\bibnamefont {Wang}}, \bibinfo {author} {\bibfnamefont
  {K.}~\bibnamefont {Cui}}, \bibinfo {author} {\bibfnamefont {H.-L.}\
  \bibnamefont {Yin}}, \bibinfo {author} {\bibfnamefont {N.-L.}\ \bibnamefont
  {Liu}}, \bibinfo {author} {\bibfnamefont {L.}~\bibnamefont {Li}}, \bibinfo
  {author} {\bibfnamefont {X.}~\bibnamefont {Ma}}, \bibinfo {author}
  {\bibfnamefont {J.~S.}\ \bibnamefont {Pelc}}, \bibinfo {author}
  {\bibfnamefont {M.~M.}\ \bibnamefont {Fejer}}, \bibinfo {author}
  {\bibfnamefont {Q.}~\bibnamefont {Zhang}}, \ and\ \bibinfo {author}
  {\bibfnamefont {J.-W.}\ \bibnamefont {Pan}},\ }\href@noop {} {\bibfield
  {journal} {\bibinfo  {journal} {Arxiv preprint arXiv:1209.6178}\ } (\bibinfo
  {year} {2012})}\BibitemShut {NoStop}%
\bibitem [{\citenamefont {Tang}\ \emph
  {et~al.}(2013{\natexlab{b}})\citenamefont {Tang}, \citenamefont {Liao},
  \citenamefont {Xu}, \citenamefont {Qi}, \citenamefont {Qian},\ and\
  \citenamefont {Lo}}]{Tang:MDIQKDexp:2013}%
  \BibitemOpen
  \bibfield  {author} {\bibinfo {author} {\bibfnamefont {Z.}~\bibnamefont
  {Tang}}, \bibinfo {author} {\bibfnamefont {Z.}~\bibnamefont {Liao}}, \bibinfo
  {author} {\bibfnamefont {F.}~\bibnamefont {Xu}}, \bibinfo {author}
  {\bibfnamefont {B.}~\bibnamefont {Qi}}, \bibinfo {author} {\bibfnamefont
  {L.}~\bibnamefont {Qian}}, \ and\ \bibinfo {author} {\bibfnamefont {H.-K.}\
  \bibnamefont {Lo}},\ }\href@noop {} {\bibfield  {journal} {\bibinfo
  {journal} {arXiv preprint arXiv:1306.6134}\ } (\bibinfo {year}
  {2013}{\natexlab{b}})}\BibitemShut {NoStop}%
\bibitem [{\citenamefont {Hwang}(2003)}]{Hwang:Decoy:2003}%
  \BibitemOpen
  \bibfield  {author} {\bibinfo {author} {\bibfnamefont {W.-Y.}\ \bibnamefont
  {Hwang}},\ }\href {\doibase 10.1103/PhysRevLett.91.057901} {\bibfield
  {journal} {\bibinfo  {journal} {Phys. Rev. Lett.}\ }\textbf {\bibinfo
  {volume} {91}},\ \bibinfo {pages} {057901} (\bibinfo {year}
  {2003})}\BibitemShut {NoStop}%
\bibitem [{\citenamefont {Lo}\ \emph {et~al.}(2005)\citenamefont {Lo},
  \citenamefont {Ma},\ and\ \citenamefont {Chen}}]{Lo:Decoy:2005}%
  \BibitemOpen
  \bibfield  {author} {\bibinfo {author} {\bibfnamefont {H.-K.}\ \bibnamefont
  {Lo}}, \bibinfo {author} {\bibfnamefont {X.}~\bibnamefont {Ma}}, \ and\
  \bibinfo {author} {\bibfnamefont {K.}~\bibnamefont {Chen}},\ }\href@noop {}
  {\bibfield  {journal} {\bibinfo  {journal} {Phys.~Rev.~Lett.~}\ }\textbf
  {\bibinfo {volume} {94}},\ \bibinfo {pages} {230504} (\bibinfo {year}
  {2005})}\BibitemShut {NoStop}%
\bibitem [{\citenamefont {Wang}(2005)}]{Wang:Decoy:2005}%
  \BibitemOpen
  \bibfield  {author} {\bibinfo {author} {\bibfnamefont {X.-B.}\ \bibnamefont
  {Wang}},\ }\href@noop {} {\bibfield  {journal} {\bibinfo  {journal}
  {Phys.~Rev.~Lett.~}\ }\textbf {\bibinfo {volume} {94}},\ \bibinfo {pages}
  {230503} (\bibinfo {year} {2005})}\BibitemShut {NoStop}%
\bibitem [{\citenamefont {Fung}\ \emph {et~al.}(2010)\citenamefont {Fung},
  \citenamefont {Ma},\ and\ \citenamefont {Chau}}]{Fung:Finite:2010}%
  \BibitemOpen
  \bibfield  {author} {\bibinfo {author} {\bibfnamefont {C.-H.~F.}\
  \bibnamefont {Fung}}, \bibinfo {author} {\bibfnamefont {X.}~\bibnamefont
  {Ma}}, \ and\ \bibinfo {author} {\bibfnamefont {H.~F.}\ \bibnamefont
  {Chau}},\ }\href {\doibase 10.1103/PhysRevA.81.012318} {\bibfield  {journal}
  {\bibinfo  {journal} {Phys. Rev. A}\ }\textbf {\bibinfo {volume} {81}},\
  \bibinfo {pages} {012318} (\bibinfo {year} {2010})}\BibitemShut {NoStop}%
\bibitem [{\citenamefont {Ma}\ \emph {et~al.}(2011)\citenamefont {Ma},
  \citenamefont {Fung}, \citenamefont {Boileau},\ and\ \citenamefont
  {Chau}}]{MXF:Finite:2011}%
  \BibitemOpen
  \bibfield  {author} {\bibinfo {author} {\bibfnamefont {X.}~\bibnamefont
  {Ma}}, \bibinfo {author} {\bibfnamefont {C.-H.~F.}\ \bibnamefont {Fung}},
  \bibinfo {author} {\bibfnamefont {J.-C.}\ \bibnamefont {Boileau}}, \ and\
  \bibinfo {author} {\bibfnamefont {H.}~\bibnamefont {Chau}},\ }\href {\doibase
  DOI: 10.1016/j.cose.2010.11.001} {\bibfield  {journal} {\bibinfo  {journal}
  {Computers \& Security}\ }\textbf {\bibinfo {volume} {30}},\ \bibinfo {pages}
  {172 } (\bibinfo {year} {2011})}\BibitemShut {NoStop}%
\bibitem [{\citenamefont {Gottesman}\ \emph {et~al.}(2004)\citenamefont
  {Gottesman}, \citenamefont {Lo}, \citenamefont {L\"utkenhaus},\ and\
  \citenamefont {Preskill}}]{GLLP:2004}%
  \BibitemOpen
  \bibfield  {author} {\bibinfo {author} {\bibfnamefont {D.}~\bibnamefont
  {Gottesman}}, \bibinfo {author} {\bibfnamefont {H.-K.}\ \bibnamefont {Lo}},
  \bibinfo {author} {\bibfnamefont {N.}~\bibnamefont {L\"utkenhaus}}, \ and\
  \bibinfo {author} {\bibfnamefont {J.}~\bibnamefont {Preskill}},\ }\href@noop
  {} {\bibfield  {journal} {\bibinfo  {journal} {Quantum Inf.~Comput.}\
  }\textbf {\bibinfo {volume} {4}},\ \bibinfo {pages} {325} (\bibinfo {year}
  {2004})}\BibitemShut {NoStop}%
\bibitem [{\citenamefont {Caves}\ \emph {et~al.}(2002)\citenamefont {Caves},
  \citenamefont {Fuchs},\ and\ \citenamefont {Schack}}]{Caves:deFinetti:2002}%
  \BibitemOpen
  \bibfield  {author} {\bibinfo {author} {\bibfnamefont {C.~M.}\ \bibnamefont
  {Caves}}, \bibinfo {author} {\bibfnamefont {C.~A.}\ \bibnamefont {Fuchs}}, \
  and\ \bibinfo {author} {\bibfnamefont {R.}~\bibnamefont {Schack}},\
  }\href@noop {} {\bibfield  {journal} {\bibinfo  {journal} {J. Math. Phys.}\
  }\textbf {\bibinfo {volume} {43}},\ \bibinfo {pages} {4537} (\bibinfo {year}
  {2002})}\BibitemShut {NoStop}%
\bibitem [{\citenamefont {Christandl}\ \emph {et~al.}(2009)\citenamefont
  {Christandl}, \citenamefont {K\"{o}nig},\ and\ \citenamefont
  {Renner}}]{Renner:deFinetti:2009}%
  \BibitemOpen
  \bibfield  {author} {\bibinfo {author} {\bibfnamefont {M.}~\bibnamefont
  {Christandl}}, \bibinfo {author} {\bibfnamefont {R.}~\bibnamefont
  {K\"{o}nig}}, \ and\ \bibinfo {author} {\bibfnamefont {R.}~\bibnamefont
  {Renner}},\ }\href@noop {} {\bibfield  {journal} {\bibinfo  {journal} {Phys.
  Rev. Lett.}\ }\textbf {\bibinfo {volume} {102}},\ \bibinfo {pages} {020504}
  (\bibinfo {year} {2009})}\BibitemShut {NoStop}%
\end{thebibliography}%

\end{document}